%% file: chapter.tex
\documentclass[graybox]{svmult}
\usepackage{mathptmx}       % selects Times Roman as basic font
\usepackage{helvet}         % selects Helvetica as sans-serif font
\usepackage{courier}        % selects Courier as typewriter font
\usepackage{type1cm}        % activate if the above 3 fonts are
\usepackage{makeidx}         % allows index generation
\usepackage{graphicx}        % standard LaTeX graphics tool
\usepackage{multicol}        % used for the two-column index
\usepackage[bottom]{footmisc}% places footnotes at page bottom
\makeindex
\usepackage{amsmath,amssymb}%amsthm

\usepackage[ansinew]{inputenc}
\usepackage[hyphens]{url}
\usepackage[dvips,breaklinks=true,colorlinks]{hyperref}

\usepackage[style=authoryear,
            url=false,
            maxbibnames=99,
            firstinits=true, % vornamen = erste buchstabe
            isbn=false,
            natbib=true,
            hyperref=true,
            bibencoding=utf8]{biblatex}
\AtEveryBibitem{\clearfield{note}}   % damit mein komentare nicht geprintet werden.

\bibliography{library.bib,jabref_philipp_utf2.bib}

\DeclareCiteCommand{\cite}
  {\usebibmacro{prenote}}
  {\usebibmacro{citeindex}%
   \printtext[bibhyperref]{\usebibmacro{cite}}}
  {\multicitedelim}
  {\usebibmacro{postnote}}

\DeclareCiteCommand*{\cite}
  {\usebibmacro{prenote}}
  {\usebibmacro{citeindex}%
   \printtext[bibhyperref]{\usebibmacro{citeyear}}}
  {\multicitedelim}
  {\usebibmacro{postnote}}

\DeclareCiteCommand{\parencite}[\mkbibparens]
  {\usebibmacro{prenote}}
  {\usebibmacro{citeindex}%
    \printtext[bibhyperref]{\usebibmacro{cite}}}
  {\multicitedelim}
  {\usebibmacro{postnote}}

\DeclareCiteCommand*{\parencite}[\mkbibparens]
  {\usebibmacro{prenote}}
  {\usebibmacro{citeindex}%
    \printtext[bibhyperref]{\usebibmacro{citeyear}}}
  {\multicitedelim}
  {\usebibmacro{postnote}}

\DeclareCiteCommand{\footcite}[\mkbibfootnote]
  {\usebibmacro{prenote}}
  {\usebibmacro{citeindex}%
  \printtext[bibhyperref]{ \usebibmacro{cite}}}
  {\multicitedelim}
  {\usebibmacro{postnote}}

\DeclareCiteCommand{\footcitetext}[\mkbibfootnotetext]
  {\usebibmacro{prenote}}
  {\usebibmacro{citeindex}%
   \printtext[bibhyperref]{\usebibmacro{cite}}}
  {\multicitedelim}
  {\usebibmacro{postnote}}

\DeclareCiteCommand{\textcite}
 {\boolfalse{cbx:parens}}
  {\usebibmacro{citeindex}%
    \printtext[bibhyperref]{\printnames{labelname}%
      \printtext{ (\printfield{year}\printtext{)}}}}
 {\ifbool{cbx:parens}
  {\bibcloseparen\global\boolfalse{cbx:parens}}
  {}%
 \multicitedelim}
{\usebibmacro{textcite:postnote}}
\inputencoding{utf8}

\newcommand{\C}{\mathbb{C}}

\newcommand{\supp}{\text{supp}}

\newcommand{\Complex}{\mathbb{C}}
\newcommand{\Reals}{\mathbb{R}}
\newcommand{\Span}{\text{span}}

\newcommand{\Order}{\mathcal{O}}
\newcommand{\Normal}{\mathcal{N}}
\newcommand{\rank}{\text{rank}}

\newcommand{\rk}{\text{rk}}
\newcommand{\proj}{\pi}
\newcommand{\ostar}{\circledcirc}
\newcommand{\vc}{\text{vec}}
\newcommand{\mysubsubsec}[1]{\runinhead{#1:}}
\newcommand{\Ell}[1]{\ell_{#1}}

\ifx\property\undefined
\newtheorem{property} 
\fi

\newcommand{\extref}[1]{}
\input{HeaderShamaiElsevier2.tex}                          % Alle Format und Extras f{\"u}r Shamai doku

\renewcommand{\cite}{\citep}
\begin{document}
\title*{Sparse Model Uncertainties in Compressed Sensing with Application to Convolutions and Sporadic Communication}
\titlerunning{Sparse Model Uncertainties in CS and Application to Sparse Convolutions}
\author{Peter Jung and Philipp Walk} \institute{Peter Jung \at Technische Universität Berlin, Straße des 17. Juni 136,
  10587 Berlin, Germany, \email{peter.jung@mk.tu-berlin.de} \and Philipp Walk \at Technische Universität München,
  Arcistrasse 21, 80333 München, Germany, \email{philipp.walk@tum.de} } \maketitle

\abstract{
  The success of the compressed sensing paradigm has shown that  a 
  substantial reduction in sampling and storage complexity can be achieved in certain linear and non--adaptive  
  estimation problems. It is
  therefore an advisable strategy for noncoherent information retrieval in, for example,
  sporadic blind and semi--blind communication and sampling problems.
  But, the conventional model is not practical
  here since the compressible signals have to be estimated from samples taken solely on
  the output of an un--calibrated system which is unknown during measurement but often compressible. 
  Conventionally, one has either to operate at suboptimal sampling rates 
  or the recovery performance substantially suffers from the dominance of model mismatch.\\
  In this work we discuss such type of estimation problems and we focus on \emph{bilinear inverse problems}.
  We link this problem to the recovery of low--rank and sparse matrices and establish stable low--dimensional 
  embeddings of the uncalibrated receive signals
  whereby addressing also efficient communication--oriented  methods like \emph{universal} random demodulation. 
  Exemplary, we investigate in more detail sparse convolutions serving as a basic communication channel model.
  In using some recent results from additive combinatorics  we show that such type of signals can be 
  efficiently low-rate 
  sampled
  by semi--blind methods. Finally, we present a further application of these results in the field of phase retrieval
  from intensity Fourier measurements.
  % \keywords{}
}

%\tableofcontents
\section{Introduction}\label{sec:intro}

Noncoherent compressed reception of information is a promising approach to cope with several future challenges in
sporadic communication where short compressible messages have to be communicated in an unsynchronized manner over unknown, but
compressible, dispersive channels.  To enable such new communication concepts efficiently, it is therefore necessary to
investigate blind and semi--blind sampling strategies which  explicitly account for the low--dimensional structure of
the signals.  Since the compressed sensing paradigm provides a substantial reduction in sampling and storage complexity
it is therefore also an advisable strategy for noncoherent information retrieval.  However, in this and many related
application the conventional linear estimation model is a quite strong assumption since here the compressible signals of
interest are not accessible in the usual way.  Instead they have to be estimated from sampling data taken solely on the
output of an additional linear system which is itself unknown during measurement but often compressible. Thus, in the
standard scheme one has either to operate at suboptimal rates or the overall estimation performance
substantially suffers from the dominance of model mismatch.  It is therefore important to evaluate the additional amount
of sampling which is necessary to cope in a stable way with such model uncertainties. The output signals to be sampled
do not constitute anymore a fixed finite union of low--dimensional canonical subspaces but a more complicated but still
compressible set.  In this chapter we focus on \emph{bilinear models} and we discuss conditions which ensure additive
complexity in input signals and model uncertainty.  We motivate the relevance of this topic for sporadic
communication in future cellular wireless networks and its random access strategies. In this setting the main dispersion
is caused by convolutions of $s$--sparse channel impulse  responses $x$ with $f$--sparse user signals $y$. The
convolution $x\ast y$ in dimension $n$ can be recovered by conventional compressed sensing methods from $\Order(sf\log(n))$ incoherent samples
whereby only $s+f$ ''active'' components contribute.  However, we will show that for fixed $x$ ordinary (non--circular) convolutions are
invertible in $y$ (and vice-versa) in a uniformly stable manner and can be compressed  into
$2^{2(s+f-2)\log(s+f-2)}$
dimensions \emph{independent} of $n$.  This demonstrates the possibility of low--rate sampling strategies in the order
$\Order((s+f)\log(n))$ in our setting. Although efficient recovery algorithms operating at this rate are still unknown
we show that sampling itself can be achieved efficiently with a considerable derandomized and universal approach, 
with a \emph{random demodulator}.
This proceeding contribution contains material from the joint work
of the authors, presented in two talks at the CSA13 workshop,
i.e.  \emph{''Low--Complexity Model Uncertainties in Compressed Sensing with Application to Sporadic Communication''} by
Peter Jung and \emph{''Stable Embedding of Sparse Convolutions''} by Philipp Walk.

\mysubsubsec{Outline of the Work} First, we state in Section \ref{sec:intro} 
the bilinear sampling problem and we discuss the relevance of this
topic for sporadic communication in future cellular wireless networks.  In
Section \ref{sec:stable} we will present a general framework for stable random low--dimensional embedding of the signal
manifolds beyond the standard linear vector model.  We discuss structured measurements in this context and propose a
\emph{universal random demodulator} having an efficient implementation.
At the end of this section we summarize in Theorem \ref{lemma:randomsampling:bilinear}
that additive scaling in sampling
complexity can be achieved for certain bilinear inverse problems, once a particular stability condition is fulfilled
independent of the ambient dimension. 
In Section \ref{sec:conv:rnmp} we will discuss such a condition for sparse convolutions in more detail and
we show in Theorem \ref{thm:dryi} by arguments from additive combinatorics that ambient dimension will not occur 
in this case.
Finally, we show a further
application for quadratic problems and draw the link to our work \cite{walk:symfourier13}
on complex phase retrieval from intensity measurements of symmetrized Fourier measurements and 
presenting this result in  Theorem \ref{thm:phaseretrieval}.

\newcommand{\Meas}{\Phi}
\newcommand{\Dict}{\Psi}
\newcommand{\Obs}{b} % observation vector
\newcommand{\Inx}{x} % first input vector
\newcommand{\Iny}{y} % first input vector
\newcommand{\Noi}{e} % noise vector
\newcommand{\Bz}{z} % z=B(x,y)

\subsection{Problem Statement}
The standard linear model in compressed sensing is that the noisy observations $\Obs=\Meas\Dict\Iny+\Noi$ are obtained
from a \emph{known} model under the additional assumption that $\Iny$ is essentially concentrated on a few components in
a fixed basis $\Dict$. Let us assume for the following exposition, that $\Iny\in\Sigma_f$ is an $f$--sparse vector.
$\Meas$ is the, possibly random, measurement matrix and $\Dict$ denotes the dictionary (not necessarily a basis) in which the object can be
sparsely described, both have to be known for decoding. For our purpose it is not important to understand
$\Dict$ as a property of the data. Instead, $\Dict$ can also be understood as a part of the measurement process, i.e.
viewing $\Meas\Dict$ as the overall measurement matrix.  Solving for a sparse parameter vector $\Iny$ in this case can
be done with a substantially reduced number of incoherent measurements.  However what happens if $\Dict$ (or $\Meas$) is \emph{not perfectly
known}, i.e. depends on some unknown parameters resulting in an overall uncertainty in the estimation model ?  To our
knowledge, this is one of the most sensible points for the application of compressed sensing to practical problems.

\mysubsubsec{Model Uncertainties}
Additive uncertainties in the overall measurement process have been investigated for example by
\cite{Herman:generaldeviants}. An extension of this work with explicit distinction between errors in $\Meas$ and
$\Dict$, suitable for redundant dictionaries, has been undertaken in \cite{Aldroubi2012}.  Another situation, referring
more to the multiplicative case, is the basis mismatch as has been studied for example by \cite{Chi:sensitivity}. The
strategy in the previous work was to estimate the degradation of the recovery performance in terms of the perturbation.
However, if the unknown uncertainty is itself compressible in some sense one might treat it as a further unknown
variable to be estimated from the same (blind) or prior (semi--blind, without calibrating the sampling device)  observations as well.  
For example, can one handle the case where $\Dict$ is
known to have a compressible representation $\Dict=\sum_j \Inx_j \Dict_j$ such that for example the coefficient vector
$\Inx\in\Sigma_s$ is $s$--sparse:
\begin{equation}
   \Obs=\Phi(\sum_{j} \Inx_j \Dict_j)\Iny+\Noi=:\Meas B(\Inx,\Iny)+\Noi
   \label{eq:model:uncertainty}
\end{equation}
In principle, the original goal is here to estimate the sparse signal vector $\Iny$ from $\Obs$ under the condition that
$\Inx$ is sparse as well. In this setting it would only be necessary to infer on the support of $\Inx$. On the
other hand, in many applications more precise knowledge on the model parameters $\Inx$ are desirable as well
and the task is then to recover the pair $(\Inx,\Iny)$ up to indissoluble ambiguities.

\mysubsubsec{Sampling Methods for Sporadic Communication}
Our motivation for investigating this problem are universal sampling methods, which may become relevant for sporadic
communication scenarios, in particular in wireless cellular networks.  Whereby voice telephone calls and human generated
data traffic were the main drivers for 2/3/4G networks (GSM, UMTS and LTE) this is expected to change dramatically in
the future.  Actually, 5G will bring again a completely new innovation cycle with many completely new and challenging
applications (see for example \cite{Wunder:COMMAG:5G} and references therein).  The \emph{Internet of Things} will
connect billions of smart devices for monitoring, controlling and assistance in, for example, the tele-medicine area,
smart homes and smart factory etc.  In fact, this will change the internet from a human-to-human interface towards a
more general machine-to-machine platform. However, machine-to-machine traffic is completely sporadic in nature and much
more as providing sufficient bandwidth.

A rather unexplored field here is the \emph{instantaneous} joint estimation of user activity, channel coefficients and data messages. As
indicated, e.g., in \cite{Dhillon13} such approaches are necessary for one--stage random access protocols and therefore
key enablers for machine--type communication within the vision of the ''Internet of Things''. For a brief exposition,
let us focus only on the estimation of a single channel vector $\Inx$ and data vector $\Iny$ from a single or only few observation cycles
$\Obs$. This vector $\Obs$ represents the samples taken by the receiver on elements $B(\Inx,\Iny)$ from a bilinear set under sparsity or more general compressibility constraints.  A
typical (circular) channel model for $B$ is obtained in \eqref{eq:model:uncertainty} with
unitary representations of the finite Weyl--Heisenberg group on, e.g., $\Complex^{n}$:
\begin{equation}
  (\Dict_j)_{kl}=e^{i2\pi\frac{j_1l}{n}}\delta_{k,l\ominus j_2}\quad\text{for}\quad j=(j_1,j_2)\in\{0,\dots, n-1\}^2.
  \label{eq:model:weylheisenberg}
\end{equation}
These $n^2$ unitary operators fulfill the Weyl commutation rule and cyclically ($\ominus$ denotes subtraction modulo
$n$) shift the data signal $y$ by $j_2$ and its discrete Fourier transform by $j_1$.  Even more they form an operator
(orthonormal) basis with respect to the Hilbert--Schmidt inner product, i.e., \emph{every} channel (linear mapping on $y$) can
be represented as $\sum_{j} x_j\Dict_j$ (spreading representation of a ''discrete pseudo--differential operator'').
Since $B$ is dispersive in the standard and the Fourier basis such channels are called doubly--dispersive and in most of the
applications here the spreading function $x$ is sparse (or compressible).  Furthermore, at moderate mobility between
transmitter and receiver $x$ is essentially supported only on $j\in\{0\}\times\{0,\dots, n-1\}$, hence, the dominating
single--dispersive effect is the \emph{sparse circular convolution}, see \secref{sec:conv:rnmp}.

For a \emph{universal} dimensioning of, for example, a future  random access channel architecture, 
where ''universal'' means that sampling strategies $\Meas$ should be independent of the 
particular low--dimensional structure of $\Inx$ and $\Iny$, it is important to know how many samples
$m$ have to be taken in an efficient manner for stable distinguishing:
\begin{itemize}
\item[(i)] \ \ $B(\Inx,\Iny)$ from $B(\Inx,\Iny')$ and from $B(\Inx',\Iny)$ by 
   universal measurements not depending on the low--dimensional structure and
   on $x$ or $y$ (\emph{semi--blind methods})\\
   \label{item:task:1}
\item[(ii)] \ \ completely different elements $B(\Inx,\Iny)$ and $B(\Inx',\Iny')$
   by universal measurements not depending on the low--dimensional structure (\emph{blind methods})
   \label{item:task:2}
\end{itemize}
In the view of the expected system parameters like $(n,s,f)$ 
there will be substantial difference in whether a \emph{multiplicative} $m=\Order(sf\log n)$ or \emph{additive} 
$m=\Order((s+f)\log n)$ scaling can be achieved. Even more, we argue that achieving additive scaling 
without further compactness assumptions is closely related to the question whether $(\Inx,\Iny)$ can be deconvolved 
up to indissoluble ambiguities at all from $B(\Inx,\Iny)$ which is in some cases also known as blind deconvolution.
That this is indeed possible in a suitable random setting has already been demonstrated for the non--sparse and sufficiently 
oversampled case in \cite{Ahmed:2012}.

\subsection{Bilinear Inverse Problems with Sparsity Priors}
Here we consider the model \eqref{eq:model:uncertainty} from a compressive 
viewpoint which means that, due the low
complexity of signal sets, the measurement matrix $\Meas\in\Complex^{m\times n}$ corresponds to undersampling
$m\ll n$.  We use the well--known approach of lifting the bilinear map
$B:\Complex^{n_1}\times\Complex^{n_2}\rightarrow\Complex^{n}$ to a linear map $B:\C^{n_1\times
n_2}\rightarrow\Complex^{n}$.  Hereby, we can understand $x\otimes y=xy^T=x\bar{y}^*$ as a complex rank--one $n_1\times
n_2$--matrix or as a $n_1\cdot n_2$--dimensional complex vector $\vc(x\otimes y)$. As long as there arise no
confusion we will always use the same symbol $B$, i.e.,  the structured signal $\Bz$ to be sampled in a compressive
manner can be written in several ways:
\begin{equation}
   \Bz=B(\Inx,\Iny)=B(\Inx\otimes \Iny)=B(\vc(\Inx\otimes \Iny)).
\end{equation}
A step--by--step approach would be (i) estimating $\Bz$ from $m$ noisy observations $\Obs=\Meas\Bz+e$ and (ii)
deconvolving $(\lambda\Inx,\Iny/\lambda)$ from that estimate up to a scaling $\lambda\neq 0$ 
(due to the bilinearity) and, depending on $B$, further ambiguities.
The second step requires injectivity of $B$ on the desired subset in $\Complex^{n_1 \cdot n_2}$  and full inversion
requires obviously $n_1 \cdot n_2\leq n$.  Both
steps usually fall into the category of inverse problems and here we will consider the case with sparsity priors on
$\Inx$ and $\Iny$.  For $\Inx\in\Sigma_s$ and $\Iny\in\Sigma_f$ the vector $\vc(\Inx\otimes \Iny)$ is $sf$--sparse
in $n_1 \cdot n_2$ dimensions,
i.e.  $\Bz=B(\Inx,\Iny)$ is the image of the $sf$--sparse vector $\vc(\Inx\otimes \Iny)\in\Sigma_{sf}$
under $B$.  For
recovering $\Bz$ (step (i)) via convex methods one could use the framework in \cite{Candes:DRIP}:
\begin{equation}
   \min \lVert B^*\Bz\rVert_{1}\quad\text{s.t.}\quad\lVert \Meas\Bz-\Obs\rVert_{2}\leq\epsilon.
   \label{eq:bpdn:analysissparsity}
\end{equation}
This (analysis--sparsity) approach recovers $\Bz$ successfully (or within the usually desired error scaling) if $\Phi$ acts
almost--isometrically on $B(\Sigma_{sf})$ (called D--RIP in \cite{Candes:DRIP}).  For example, if $\Meas$ obeys a
certain concentration bound (like i.i.d. Gaussian matrices), $m=\Order(sf\log(n_1\cdot n_2/(sf)))$ and $B$ is a
partial isometry ($BB^*$ is a scaled identity, i.e. the columns of $B$ form a tight frame)
\eqref{eq:bpdn:analysissparsity} succeeds with exponential probability.  Once $B$ would be injective on $\Sigma_{sf}$ it
is in principal possible to extract $\vc(\Inx\otimes\Iny)$ from $\Bz$ (step (ii)). In this case one could also consider
directly the (synthesis--sparsity) approach for recovery, i.e., minimizing for example
$\Ell{1}$--norm over the vectors $u \in\Complex^{n_1 \cdot n_2}$:
\begin{equation}
   \min \lVert u\rVert_{1}\quad\text{s.t.}\quad\lVert \Meas B(u)-\Obs\rVert_{2}\leq\epsilon.
   \label{eq:bpdn:synthesissparsity}
\end{equation}
This one--step approach in turn depends on the precise mapping properties of $B$ (in particular its anisotropy) 
and a detailed characterization could be done in terms of the RE--condition in \cite{Bickel2009}.
For  $B$ being unitary \eqref{eq:bpdn:synthesissparsity} agrees with \eqref{eq:bpdn:analysissparsity}. Another 
approach in estimating the RIP--properties of the composed matrix $\Meas B$ for random measurements $\Meas$ having
the concentration property was given in \cite{Rauhut2008:anisotropic} yielding successful recovery
in the regime $m=\Order(sf\log(n_1\cdot n_2/(sf)))$.

\newcommand{\slmat}{M} % sparse and low--rank matrices
But the set $\Sigma_{sf}$ is considerable larger than $\vc(\slmat_{s,f})$ where $\slmat_{s,f}$ denotes the
rank--one tensor products, i.e.:
\begin{equation}
   \slmat_{s,f}:=\{\Inx\otimes\Iny\,:\,\Inx\in\Sigma_s\,\text{and}\,\Iny\in\Sigma_f)
   \label{eq:slmat}
\end{equation}
are the sparse rank--one $n_1\times n_2$ matrices 
with $s$ non--zero rows and $f$ non--zero columns. All the previous consideration
make only use of the vector--properties of $\Inx\otimes\Iny$ and, hence, result in a multiplicative $sf$--scaling.
Although the approach \cite{Gleichman2011} termed \emph{blind compressed sensing} gives structural insights 
into the rank--sparsity relation it also results in the $sf$--regime.  Since the non--zero entries in $\vc(\slmat_{s,f})$ occur in $s$ equal--sized
blocks, each having at most $f$ non--zero values, one might extend the vector--concept to block--sparse vectors.  However, to
fully exploit the correlation properties in the non--zero coefficients one has to consider the original estimation
problem as a low--rank matrix recovery problem with sparsity constraints as already investigated in \cite{Choudhary2014}
in the view of noiseless identifiability.
Unfortunately, already without sparsity this setting does not fall directly into the usual isotropic low--rank matrix recovery setting since 
the matrix picture is somehow ''hidden behind'' $B$ resembling the anisotropy
of $\Phi B$.  Whereby the anisotropic vector case has been extensively investigated in vector--RIP context by
\cite{Bickel2009,rudelson:anisotrop} or in the RIP'less context by \cite{Kueng12} (noiseless non--uniform recovery) very
little is known here in the matrix case.

We already mentioned that restricting the problem \eqref{eq:model:uncertainty} solely to the diagonal $B(\Inx,\Inx)$, which is 
a quadratic inverse problem, resembles closely the phase--retrieval problem. 
We will make this precise for the non--compressive case in  Section \ref{sec:conv:rnmp}. Unfortunately this does
not extends to the compressive case. Finally, we mention that our framework applies also to a certain extent to the multi--linear setting with minor
modification, i.e. for higher order tensors. Such constructions occur not only in image and video processing
but also in the communication context. 
For example, a separable spreading is characterized by 
$\Inx_{(j_1,j_2)}=\Inx^\text{\tiny (1)}_{j_1}\cdot \Inx^\text{\tiny (2)}_{j_2}$ with \eqref{eq:model:weylheisenberg}
and is widely used as a simplified channel model yielding a $3$th order inverse problem
in \eqref{eq:model:uncertainty}.

\section{Stable Low--Dimensional Embedding}
\label{sec:stable}
In this section we will establish a generic approach for stable embedding a non--compact and non--linear set
$V$ of a finite--dimensional normed space into a lower--dimensional space.  
Hereby $V\subseteq Z-Z$ will usually represent (not necessarily
all) differences (chords/secants) between elements from another set $Z$.  In this case, stable lower--dimensional 
embedding essentially establishes the existence of a ''shorter description'' of elements in $Z$ whereby decoding can be 
arbitrarily complex. Once $Z$ obeys a suitable convex surrogate function $f:Z\rightarrow\Reals_+$ the geometric
convex approach \cite{Chandrasekaran2010} is applicable and via Gordon's ''escape through a mesh'' theorem 
the Gaussian width of the descent cones of $f$ provide estimates on the sampling complexity.
But, e.g., for $Z=\slmat_{s,f}$  this seems not be the case and conventional multi--objective 
convex programs ($f$ is the infimal convolution of multiple surrogates like $\Ell{1}$ and nuclear norm) are limited to the Pareto boundary \cite{Oymak2012} formed by the single objectives and have
a considerable gap to the optimal sampling complexity.

\if0
However, if for example $V$ is a
union of difference sets 

Since $V$ might lack of linear structure we (i) first establish a 
\fi

\subsection{Structure--aware Approximation}
The first step is to establish an approximation statement in terms
of an intrinsic distance description for a given subset $V$ of 
a finite--dimensional space with a given norm $\lVert\cdot\rVert$.
The problem arises since we need to quantify certain differences $v-r$ of two elements $v,r\in V$ whereby
potentially $v-r\notin V$. Clearly $v-r\in\Span(V)$, but in this way we will 
potentially loose the low--complexity structure of $V$. 

\mysubsubsec{Intrinsic Distance}
We consider an intrinsic notion of a 
\emph{distance function} $d(v,r)$.
For example, if $V$ is path--connected, one might take the length of 
a smooth path $\gamma:[0,1]\rightarrow V$ from $v=\gamma(0)$ to $r=\gamma(1)$ 
with $\gamma([0,1])\subset V$ and use
$\lVert v- r\rVert\leq\int_0^1\lVert\dot{\gamma}(t)\rVert dt$.
However, we will not exploit Riemannian metrics here as has been done, for example,
in \cite{Baraniuk09:manifolds} for \emph{compact} manifolds.
Instead, since 
a \emph{rectifiable} path can be approached by finite partition sums,
we consider a construction called \emph{projective norm} in the case of tensors, see
here \cite[p.7]{Diestel2008} or \cite[Ch.2]{Ryan:tensorproducts}. 
More precisely, with $w=v-r\in V- V$, this means:
\begin{equation}
   \lVert w\rVert_\proj:=\inf\{\sum_{i}\lVert v_i\rVert\,: w=\sum_i v_i\,\text{with}\, v_i\in V\},
   \label{eq:approx:path:sum}
\end{equation}
whereby for any $v\in V$ one has  $\lVert v\rVert_\proj=\lVert v\rVert$. If
there is no decomposition for $w$ then $\lVert w\rVert_\proj$ is set to $\infty$.
In the examples later on we will always have that $V$ is a central--symmetric linear cone, i.e.
$V=\xi V$ for all $0\neq \xi\in\mathbb{R}$. In this case $V$ is generated by 
a central-symmetric atomic subset of the unit sphere and $\lVert w\rVert_\proj$ is then 
called an \emph{atomic norm}, see here for example \cite{Chandrasekaran2010}. 
%A proof that $\lVert w\rVert_\proj$ has the
%properties of a norm can be taken from \cite[Prop. 2.1]{Ryan:tensorproducts}.  
But, instead of approaching the optimum in \eqref{eq:approx:path:sum},  we will later consider a particularly chosen
decomposition $v-r=\sum_i v_i$ depending from the application
(we specify the relevant cases later in Section \ref{subsec:coveringandentropy}).
Once, for $v-r$ a decomposition $\{v_i\}$ in $V$ has been specified, $d(v,r)$ is defined 
(and lower bounded) as:
%\lVert v-r\rVert\leq\lVert v-r\rVert_\proj\leq\sum_{i}\lVert v_i\rVert=:d(v,r).
\begin{equation}
   \begin{split}
      d(v,r):=\sum_{i}\lVert v_i\rVert\geq\lVert v-r\rVert_\proj\geq\lVert v-r\rVert
   \end{split}
   \label{eq:approx:d:def}
\end{equation}
However, then $d(v,r)$ is not necessarily a metric on $V$.
There is a useful condition: if $v-r$ has a $k$--term decomposition 
$v-r=\sum_{i=1}^{k}v_i$ in $V$ and there is $\mu$ such that
$\sum_{i=1}^k\lVert v_i\rVert^2\leq \mu\lVert\sum_{i=1}^k v_i \rVert^2$  -- it would follow 
from Cauchy--Schwartz inequality that:
\begin{equation}
   \begin{split}
      \lVert v-r\rVert\leq d(v,r)=
      \sum_{i=1}^k \lVert v_i\rVert
      \leq (k\sum_{i=1}^k \lVert v_i\rVert^2)^{\frac{1}{2}}
      \leq\sqrt{k\mu}\lVert v-r\rVert.
   \end{split}
   \label{eq:approx:splitting}
\end{equation}
Thus, within $\sqrt{k\mu}$ the norms $\lVert\cdot\rVert_\proj$ and $\lVert\cdot\rVert$ are then equivalent on $V$
and for Euclidean norms we have $\mu=1$ for orthogonal decompositions. 
A worst--case estimate is obviously $k=\dim(V)$ meaning 
that $V$ contains a frame for its span with lower frame--bound $1/\mu>0$ which, however,
could be arbitrary small depending on the anisotropic structure of $V$.
Instead, we shall therefore consider $V=B(U)$ as the image under a given mapping 
$B$ of another ''nicer'' set $U$ which a-priori has this property.
%\footnote{We have to check:
%  (0) $1$-summing operators on non-linear spaces;
%  (1) Banach space of cotype $2$; (2) Decomposition of frames with good bound (Kadinson-Singer and
%  the paving conjecture due to Weaver)}.

\mysubsubsec{The Essential Approximation Step}
\newcommand{\epsV}{\epsilon}
\newcommand{\epsU}{{\hat\epsilon}}
Here we will now give a short but generalized form of the essential step in \cite{Baraniuk2008}.  Variants of it can be
found in almost all RIP--proofs based on nets. However, here we focus on the important property, that the 
(linear) compression mapping $\Phi:V\rightarrow W$ should be solely applied on elements of $V$.  
\if0
\begin{lemma}
   Fix $\del\in(0,1)$ and $0<\epsV<1$. Let $\Phi:V\rightarrow W$ be a linear map between subsets $V$ and $W$ of finite normed spaces, 
   each with its norm $\lVert\cdot\rVert$.
   Assume that for each $v\in V$ we have (i) $r=r(v)\in V$ and  (ii) a decomposition
   $v-r=\sum_{i} v_i$ with $\{v_i\}\subset V$ such that 
   $d(v,r):=\sum_{i}\lVert v_i\rVert\leq\epsV\lVert v\rVert$ and (iii)
   $|\lVert\Phi r\rVert-\lVert r\rVert|\leq\frac{\delta}{2}\lVert r\rVert$.
   Then it also holds:
   \begin{equation}
      |\lVert\Phi v\rVert-\lVert v\rVert|\leq\delta\lVert v\rVert\quad\text{for all}\quad v\in V
   \end{equation}
   if $\epsV<\delta/7$ ($\epsV<\delta/4$ if $\lVert v\rVert=\lVert r(v)\rVert$ for all $v\in V$).
   \label{lemma:intrinsic:approx}
\end{lemma}
\fi
\begin{lemma}
   Let $\del,\epsV\in(0,1)$ and $\Phi:V\rightarrow W$ be a linear map between subsets $V$ and $W$ of finite normed spaces, 
   each with its norm $\lVert\cdot\rVert$.
   Assume that for each $v\in V$ there exists $r=r(v)\in V$ such that 
   \begin{itemize}
   \item[(i)] $ $ a decomposition $\{v_i\}\subset V$ exists for $v-r=\sum_{i} v_i$ with 
   $d(v,r):=\sum_{i}\lVert v_i\rVert\leq\epsV\lVert v\rVert$ 
   \item[(ii)] $ $ and $|\lVert\Phi r\rVert-\lVert r\rVert|\leq\frac{\delta}{2}\lVert r\rVert$.
   \end{itemize}
   Then it holds for $\epsV<\delta/7$:
   \begin{equation}
      |\lVert\Phi v\rVert-\lVert v\rVert|\leq\delta\lVert v\rVert\quad\text{for all}\quad v\in V.
      \label{eq:lemma:intrinsic:approx}
   \end{equation}
   If $\lVert v\rVert=\lVert r(v)\rVert$ for all $v\in V$ then \eqref{eq:lemma:intrinsic:approx}
   holds also for $\epsV<\delta/4$. 
   \label{lemma:intrinsic:approx}
\end{lemma}
Let us make here the following remark: Lemma \ref{lemma:intrinsic:approx} neither requires that 
$V$ is symmetric ($V=-V$) nor is a linear cone ($V=\xi V$ for all $\xi>0$).
However, if the lemma holds for a given 
$V$ and approximation strategy $v\rightarrow r(v)$ then it also holds for $\xi V$ with $\xi\in\mathbb{C}$
and $r(\xi \cdot):=\xi r(\cdot)$\footnote{%
To see this, let $v,r(v)\in V$ with a decomposition $v-r(v)=\sum_{i}v_i$. Then $r(\xi v)=\xi r(v)\in\xi V$ 
and $\xi v-r(\xi v)=\xi(v-r(v))=\sum_i \xi v_i$ with $\xi v_i\in \xi V$. 
}. It holds therefore also for $\bigcup_{\xi\in\mathbb{C}}\xi V$ whereby the converse is wrong. 
\begin{proof}
   We set $a=0$ if we already know that $\lVert v\rVert=\lVert r\rVert$ and $a=1$ else.
   Using triangle inequalities we get for $v,r\in V$ with 
   the decomposition $v-r=\sum_{i=1}^k v_i$ given in (i):
   \begin{equation}
      \begin{split}
         |\lVert\Phi v\rVert-\lVert v\rVert|
         &=|\lVert\Phi v\rVert-\lVert\Phi r\rVert+ \lVert\Phi r\rVert-\lVert v\rVert |\\
         &\leq|\lVert\Phi v\rVert-\lVert\Phi r\rVert|+
         a|\lVert v\rVert-\lVert r\rVert|+|\lVert\Phi r\rVert-\lVert r\rVert|\\
         &\leq
         \lVert\Phi \left(v-r\right)\rVert+
         a\lVert v- r\rVert+|\lVert\Phi r\rVert-\lVert r\rVert|\\
         &\overset{\text{(ii)}}{\leq}
         \sum_i \lVert \Phi v_i\rVert+a\cdot d(v,r)+\frac{\delta}{2}\lVert r\rVert\\
      \end{split}
      \label{eq:approx}
   \end{equation}
   where in the last step we also used the property of $d$ given in \eqref{eq:approx:d:def}.
   Since $|\lVert v\rVert-\lVert r\rVert|\leq a\lVert v-r\rVert\leq a\cdot d(v,r)$ we have
   $\lVert r\rVert\leq \lVert v\rVert+a\cdot d(v,r)$ and therefore:
   \begin{equation}
      \begin{split}
         |\lVert\Phi v\rVert-\lVert v\rVert|
         &\leq 
         \sum_i \lVert \Phi v_i\rVert+a(1+\frac{\delta}{2})\cdot d(v,r)+\frac{\delta}{2}\lVert v\rVert\\
         &\leq
         \sum_i \lVert \Phi v_i\rVert+\left(a(1+\frac{\delta}{2})\epsV+\frac{\delta}{2}\right)\lVert v\rVert\\
      \end{split}
      \label{eq:lemma:approx:eq2}
   \end{equation}
   where the last line follows from $d(v,r)\leq\epsV\lVert v\rVert$ given in the assumption (i) of the lemma.
   Note that, if we can ensure $\lVert r\rVert=\lVert v\rVert$ then $a=0$.
   We now follow the same strategy as in \cite{Baraniuk2008} and
   define the constant:
   \begin{equation}
     A:=\sup_{0\neq v\in V}\frac{|\lVert\Phi v\rVert-\lVert v\rVert|}{\lVert v\rVert}.
     \label{eq:lemma:approx:defA}
   \end{equation}
   implying that for any $\epsV'>0$ there is $v^*\in V$ 
   with $(A-\epsV')\lVert v^*\rVert\leq |\lVert\Phi v^*\rVert-\lVert v^*\rVert|$. 
   From the prerequisite (i) of the lemma there also exists $r^*=r(v^*)\in V$ with 
   $d(v^*,r^*):=\sum_i \lVert v^*_i\rVert\leq\epsV\lVert v^*\rVert$ for a decomposition
   $v^*-r^*=\sum_i v^*_i$.
   We have then from \eqref{eq:lemma:approx:defA} that
   $\sum_i \lVert \Phi v^*_i\rVert\leq (1+A) d(v^*,r^*)\leq (1+A)\epsV\lVert v^*\rVert$ and
   using \eqref{eq:lemma:approx:eq2} for $v=v^*$ gives:
   \begin{equation}
      \begin{split}
         (A-\epsV')\lVert v^*\rVert
         &\leq|\lVert\Phi v^*\rVert-\lVert v^*\rVert|
         \leq \left((1+A)\epsV+a(1+\frac{\delta}{2})\epsV+\frac{\delta}{2}\right)\lVert v^*\rVert.\\
      \end{split}
   \end{equation}
   Solving for $A$ gives:
   \begin{equation}
      \begin{split}
         A
         \leq \frac{\epsV+a(1+\frac{\delta}{2})\epsV+\frac{\delta}{2}+\epsV'}{1-\epsV}
         \overset{(!)}{\leq}\delta\,\Leftrightarrow\,
         \epsV\leq\frac{\delta-2\epsV'}{2+a(2+\delta)+2\delta}\,\Leftarrow\,
         \epsV<\frac{\delta}{4+3a},
         \\
      \end{split}
      \label{eq:A:general}
   \end{equation}
   since for each fixed $\del<1$ there exists a sufficiently small $\epsV'>0$ such that 
   \eqref{eq:A:general} holds. 
   Recall, in general, $a=1$ but if we are able to choose $\lVert r\rVert=\lVert v\rVert$ 
   we have $a=0$.\qed
\end{proof}
Summarizing, the approximation strategy of \cite{Baraniuk2008} applies 
in a quite generalized context. For a given $V$ one has (i) to
find a suitable\footnote{A suitable \emph{decomposition strategy} for $v-r=\sum_i v_i$ with all $v_i\in V$ has to 
be found. Then $d(v,r):=\sum_i \lVert v_i\rVert$ defines an intrinsic distance function.
We will give examples in Section \ref{subsec:coveringandentropy}.} $d(v,r)$ and
(ii) find covering number estimates for $V$ in terms of $d$ 
which are better than those of the ambient space. However, the second
step seems notoriously difficult and we approach this by sticking to a particular
parametrization of the set $V$.

\subsection{Bi-Lipschitz Mappings and the RNMP}
\newcommand{\eucmod}{\sigma} % modulus of euclidean
\newcommand{\propA}[1]{$(A_{#1})$}
\newcommand{\propB}[1]{$(B_{#1})$}
Here, we consider now the non--linear set $V$ as the image $V=B(U)$ 
of a (parameter) set $U$ of a normed space under a linear map
$B:U\rightarrow V$, i.e. $B$ is always a surjection. 
The domain $U$ can be, for example, subsets of vectors or matrices equipped with the norm of the
ambient space (usually some Euclidean norm).
We shall approximate each element $v\in V$ by another element $r=r(v)\in V$  
but taking care of the case $v-r\notin V$. To this end we will perform the approximation in the domain $U$ of $B$
and translate this afterwards to its range $V$.
Thus, we will need the following two properties
\propA{\eucmod} and \propB{\alpha,\beta}:\\[.5em] 
%\begin{property}
{\bf \propA{\eucmod}}: A set $U$ has the property \propA{\eucmod} for $\eucmod>0$ if
it is the finite union $U=\bigcup_{l=1}^L U_l$ of $L$ subsets
of a normed space and
for each $u,\rho\in U$ with $u,\rho\in U_l$ for some $l=1\dots L$
there exists $\{u_i\}_{i=1}^k\subset U_l$ yielding a $k$--term decomposition $u-\rho=\sum_{i=1}^k u_i$ with:
\begin{equation}
   \sum_{i=1}^k \rVert u_i\lVert\leq \eucmod\cdot\rVert \sum_{i=1}^k u_i\lVert
   \label{eq:approx:decomposition}.
\end{equation}
%\label{property:eucmod}
%\end{property}
%
%The definition of $\sigma$ has only a local character, i.e., it 
%is \emph{not} necessary to hold for any $v-r(v)\in V-V$. In particular, if
%$V=B(U)$ is a union of subsets  $v$ and $r$ can be restricted to the same subset.
For example, if $U$ is a subspace then $u-\rho\in U$ for \emph{each} $u,\rho\in U$. In this case,
the ``$k=1$''--decomposition $u_1=u-\rho$ is valid giving $\eucmod=1$.  However,
if $U$ is a union of $L$ subspaces $U_l$ then $u$ and $\rho$ usually have to be in
the \emph{same} subspace for $\eucmod=1$. 
On the other hand, if $U$ is some subset equipped with an Euclidean norm
and $u-\rho\notin U$ but is guaranteed to have an \emph{orthogonal} $k$--term decomposition in $U$
then $\eucmod=\sqrt{k}$, see \eqref{eq:approx:splitting} for $U=V$. 
For example, let $U$ be the matrices of maximal rank $\kappa$ equipped with the Frobenius (Hilbert--Schmidt) norm. 
In this case it might happen that $u-\rho\notin U$ but the singular value decomposition provides an orthogonal (in the Hilbert--Schmidt 
inner product) ``$k=2$''--decomposition in $U$ for
\emph{any} $u,\rho\in U$, i.e. $\sigma=\sqrt{2}$.
However, if $U$ is the union of $L$ matrix subsets $U_l$ of maximal rank $\kappa$ (like sparse low--rank matrices) the $u$ and $\rho$ have usually to be from the \emph{same} subset
for \eqref{eq:approx:decomposition} to hold with $\sigma=\sqrt{2}$. 
\\[.5em] 
%%%%%%%%%%%%%%%%%%%%%%%%%%%%%%%%%%%%%%%%%%%%%%%%%%%%%%%%%%%%%%%%%%5
\noindent To switch now between domain and range of $B$ we will also need the property: 
\\[.5em]
\noindent {\bf \propB{\alpha,\beta}}:
%\begin{property}
A map $B:U\rightarrow V$ has the property \propB{\alpha,\beta} if there is $0<\alpha\leq\beta<\infty$ such that it holds:
\begin{equation}   
   \alpha\lVert u\rVert\leq \lVert B(u)\rVert\leq\beta\lVert u\rVert
   \quad\text{for all}\quad u\in U   
   \label{eq:approx:rnmp}
\end{equation}
%\label{property:rnmp}
%\end{property}
%
In \cite{WJ12b} the authors have considered condition (ii) for $U=\{\Inx\otimes\Iny\,:\, \Inx\in X,\Iny\in Y\}$ where
$X$ and $Y$ are two given cones of an Euclidean space under the name \emph{restricted norm multiplicativity property}
(RNMP) since in this case $\lVert\Inx\otimes\Iny\rVert=\lVert\Inx\rVert\lVert\Iny\rVert$. We will further discuss such
models for $U$ in \eqref{eq:model:bilinear:single} and \eqref{eq:model:quadratic} below and in much more detail for
convolutions in Section \ref{sec:conv:rnmp}.  On the other hand, for  difference sets $U=M-M$ and linear mappings $B$
this is the \emph{bi-Lipschitz condition} of $B$ on $M$. We have the following lemma:
%
% \begin{definition} Let $U$ be an arbitrary subset of a normed space $X$ and $B:X\to V$ be a map. Then $B$ is called a
% \emph{bi-Lipschitz} mapping with constants $\alp,\bet>0$, if it holds % \begin{equation}   \alpha\lVert u_1
% -u_2\rVert\leq \lVert B(u_1)-B(u_2)\rVert\leq\beta\lVert u\rVert \quad,\quad u\in U.  \label{eq:bilipschitz}
%  \end{equation} \end{definition}
%
%If the mapping $B$ is linear then this condition is a linear $\alp$-embedding of $U$ in $B(U)\subset V$.  Obviously, in
%finite dimensions such $\beta<\infty$ can always be found whereby $\alpha>0$ is a strong requirement for $B$.
% 
\begin{lemma}
   Let $\epsU>0$ and $B:U\rightarrow V$ be a linear map having property
   \propB{\alpha,\beta}.    
   If $u,\rho\in U$ fulfill $\lVert u-\rho\rVert\leq\epsU\lVert u\rVert$
   and there exists a decomposition $\{u_i\}\subset U$ with $u-\rho=\sum_i u_i$ 
   such that \eqref{eq:approx:decomposition} holds for some $\eucmod>0$.
   Then it holds:
   \begin{equation}
      \lVert v-r\rVert\leq d(v,r)\leq\frac{\beta\eucmod}{\alpha}\epsU\lVert v\rVert.
   \end{equation}
   where $v:=B(u)$, $r:=B(\rho)$ and $d(v,r):=\sum_{i}\lVert B(u_i)\rVert$.
   \label{lemma:domainapprox}
\end{lemma}
%
%The essence of the lemma is, that if two points $u,\rho\in U$ are close 
%in the domain of $B$ and have a
%decomposition $u-\rho=\sum_i u_i$ in $U$ which fullfils 
%\eqref{eq:approx:decomposition} for its images $v=B(u)$ and $r=B(\rho)$ the same is true not only
%in the norm of $V$ but also with respect to $d(v,r)=\sum_i\lVert B(u_i)\rVert$.
%
\begin{proof}
   The assertion follows directly from:
   \begin{equation}
      \begin{split}
         \lVert v-r\rVert
         &=\lVert B(u)-B(\rho)\rVert=\lVert B(u-\rho)\rVert=
         \lVert\sum_{i} B(u_i)\rVert\leq\sum_{i} \lVert B(u_i)\rVert=d(v,r)\\
         &\overset{\eqref{eq:approx:rnmp}}{\leq} \beta\sum_{i} \lVert u_i\rVert
         \overset{\eqref{eq:approx:decomposition}}{\leq}\beta\eucmod\lVert u-\rho\rVert
         \leq \beta\eucmod\epsU\lVert u\rVert
         \overset{\eqref{eq:approx:rnmp}}{\leq} 
         \frac{\beta\eucmod}{\alpha}\epsU\lVert v\rVert\quad \qed
      \end{split}
      \label{eq:lemma:domainapprox:proof1}
   \end{equation}
\end{proof}
In the next section we will use this lemma to translate the accuracy in approximating
$u$ by some $\rho=\rho(u)$ from domain $U$ of $B$ to its image $V$. 
Note, that linearity of $B$ is used only in the first step
of \eqref{eq:lemma:domainapprox:proof1} whereby extension are possible once
there holds $\lVert B(u)-B(\rho)\rVert\leq c\cdot\sum_{i=1}^{k} \lVert B(u_i)\rVert$ uniformly for every 
$u\in U$ and $\rho=\rho(u)$. However, we will not further argue on this here.

\subsection{Covering and Entropy}
\label{subsec:coveringandentropy}
The remaining task is now to specify for given $\epsU>0$ an approximation strategy
$u\rightarrow \rho(u)$ such that
Lemma \ref{lemma:domainapprox} can be applied to all $u\in U$, i.e.
such that for each $u\in U$ there is $\rho=\rho(u)\in U$ with
$\lVert u-\rho\rVert\leq\epsU\lVert u\rVert$ and each $u-\rho$ has a decomposition in $U$.
From this equation it is clear that we have to consider $\epsU$--coverings\footnote{
	$R$ is an $\epsU$--net for $U'$ if for each $u\in U'$ exists $\rho=\rho(u)\in R$
	with $\lVert u-\rho\rVert\leq \epsU$, i.e. the union of these $\epsU$--balls
	centered at $\rho$ cover $U'$.
} for the set 
$U':=\{u/\lVert u\rVert\,:\, 0\neq u\in U\}$ and
to estimate its \emph{covering number}:
\begin{equation}
	N_\epsU(U'):=\min\{ |R|\,:\,\text{$R$ is an $\epsU$--covering for $U'$}\}
\end{equation}
Its logarithm $H_\epsU(U')=\log N_\epsU(U')$ is called the (metric) $\epsU$--entropy of $U'$.
Due to pre-compactness of $U'$ as a subset of the unit ball these quantities are always finite.
Furthermore we will abbreviate now $V':=\{v/\lVert v\rVert\,:\, 0\neq v\in V\}$.
Let us restate Lemma \ref{lemma:domainapprox} in this context:
\begin{corollary}
   Let $B:U\rightarrow V$ be linear with property
   \propB{\alpha,\beta}, $U$ be a linear 
   cone with property \propA{\eucmod} and $\epsU>0$. 
   Then each $\epsU$--covering for $U'$ induces an
   $\epsV=\frac{\beta\eucmod\epsU}{\alpha}$--covering for $V'$ for the norm
   in $V$ as well as for an intrinsic distance and:
   \begin{equation}
      H_{\epsV}(V')\leq H_{\alpha\epsV/(\beta\eucmod)}(U')
      \label{eq:cor:UVcovering}
   \end{equation}
   holds.
   \label{cor:UVcovering}
\end{corollary}
The property \propA{\eucmod} \emph{always} induces an intrinsic distance on $V$
as will be seen in the proof below. 
\begin{proof}
  Let be $R\subset U'$ an $\epsU$--covering for $U'$, i.e. for each $u\in U'$ there exists 
  $\rho(u)\in R$ such that 
  $\lVert u-\rho(u)\rVert\leq\epsU$. 
  Since $U$ is a linear cone, i.e. $\xi U=U$ for $\xi>0$, it follows for all $0\neq u\in U$
  that $\lVert u-\rho(u)\rVert\leq \epsU\lVert u\rVert$ holds with
  $\rho(u):=\rho(u/\lVert u\rVert)\lVert u\rVert\in U$.
  
  Property \propA{\eucmod} asserts now that there always exists a decomposition $\{u_i\}\subset U$ for 
  $u-\rho(u)=\sum_i u_i$ in $U$ satisfiying \eqref{eq:approx:decomposition}.
  For each $v:=B(u)$ set
  $r=r(v):=B(\rho(u))$ and therefore $r-v=\sum_i B(u_i)$ has an intrinsic decomposition in $V$.
  Define $d(v,r):=\sum_{i}\lVert B(u_i)\rVert$.
  From Lemma \ref{lemma:domainapprox} it follows that:
  \begin{equation*}
    \lVert v-r\rVert\leq d(v,r)\leq \frac{\beta\eucmod\epsU}{\alpha}\lVert v\rVert.
  \end{equation*}   
  Indeed, for each $v\in V'$ this 
  means $\lVert v-r\rVert\leq d(v,r)\leq \beta\eucmod\epsU/\alpha$
  which yields \eqref{eq:cor:UVcovering} and shows
  that $v$ and $r$ are also close in the intrinsic distance induced 
  by \propA{\eucmod}.\qed
\end{proof}
We will now give a short overview on some cases for $U$
which have property \propA{\eucmod}, their entropy bounds 
and the corresponding values for 
$\eucmod$ in \eqref{eq:approx:decomposition}. All examples are central--symmetric linear cones,
i.e. $U=\xi U$ for all $0\neq\xi\in\mathbb{R}$. 
Hence, Corollary \ref{cor:UVcovering} will translate this via $\epsU=\alpha\epsV/(\beta\eucmod)$
to an entropy estimate for $V$ once 
$B$ has property \propB{\alpha,\beta}.
If we assume that $U\subseteq\bigcup_{l=1}^L U_l$ we have
$N_\epsU(U')\leq \sum_{l=1}^L N_\epsU(U'_l)$ and if furthermore 
all $U'_l:=\{u/\lVert u\rVert\,:\,0\neq u\in U_l\}$ have the same covering number as $U'_1$ we get therefore:
\begin{equation}
   H_\epsU(U')\leq H_\epsU(U'_1)+\log L.
   \label{eq:covering:stdestimate}
\end{equation}
Of most interest here is the dependency on the ambient dimension of $U$. If there is sufficient compressibility the ambient
dimension will \emph{explicitly} occur only in $L$ whereby $H_\epsU(U'_1)$ could  depend on it,
for fixed $\epsV>0$, only through
$\epsU=\alpha\epsV/(\beta\eucmod)$. This is indeed the case for sparse vectors and matrices as it will be
shown now.

\newcommand{\epsScale}{\xi}
\runinhead{Finite Union of Subspaces:}
If each $U_l$ is contained in a subspace of real dimension $d$ then one can choose 
for any $\epsU>0$ and each $l=1\dots L$ an $\epsU$--net for the unit ball $\tilde{U}_l'$ in $\tilde{U}_l:=\Span(U_l)$ and one has the well--known estimate 
$H_\epsU(U'_l)\leq H_\epsU(\tilde{U}_l')\leq d\log(3/\epsU)$ being
valid for any norm not only for the Euclidean norm \cite[Sec. 2.2]{Vershynin:2010:csbook}\extref{net cardinality}.
Even more, any smooth manifold of real dimension $d$ behaves in this way for $\epsU\rightarrow 0$. 
The union of these $L$ nets is an $\epsU$--net
for $U'$. Thus, if $U$ is therefore contained in a union of $L$ subspaces of the same dimension $d$ we have from 
\eqref{eq:covering:stdestimate}:
\begin{equation}
   H_\epsU(U')\leq d\log(3/\epsU)+\log L
\end{equation}
In particular, in a subspace we have $\eucmod=1$ in \eqref{eq:approx:decomposition} as already 
explained after \eqref{eq:approx:decomposition}. 
Furthermore, in the sparse vector
case, 
$U=\Sigma_{2k}$ is the union of $L:=\binom{n}{d}\leq(\frac{en}{d})^{d}$ different $d=2k$--dimensional
subspaces and we have in this case
$H_\epsU(U')\leq d\log(3/\epsU)+d\log(en/d)$.

\mysubsubsec{Low--rank Matrices} 
Consider differences of rank--$\kappa$ matrices $M$, i.e.
$U=M-M$ are $n\times n$ matrices of rank at most $2\kappa$
with the Euclidean (Frobenius) norm 
$\lVert u\rVert^2:=\langle u,u\rangle$ defined by the Hilbert--Schmidt inner product.
From \cite[Lemma 3.1]{Candes09:LMR} it follows:
\begin{equation}
   H_\epsU(U')\leq (2n+1)2\kappa\log(9/\epsU)
   \label{eq:entropyU:lowrank}.
\end{equation}
A matrix $u-\rho$ for \emph{any} $u,\rho\in U$ has rank at most $4\kappa$ and can be decomposed as $u-\rho=u_1+u_2$ for $u_1,u_2\in U$ 
with $\langle u_1,u_2\rangle=0$, i.e. it fulfills 
\eqref{eq:approx:decomposition} for $k=2$ and $\eucmod\leq\sqrt{2}$. Hence, $U$ 
has property \propA{\eucmod} for $\eucmod=\sqrt{2}$.

\mysubsubsec{Low--rank and Sparse Matrices} 
\label{subsubsec:SLa}
Here we consider the union $U=\slmat^\kappa_{s,f}-\slmat^\kappa_{s,f}$ of 
$L=\binom{n}{2s}\binom{n}{2f}$ different sets of differences of rank-$\kappa$ matrices $\slmat^\kappa_{s,f}$
(equipped with the Frobenius norm) as defined
in \eqref{eq:slmat}
and it follows from \eqref{eq:covering:stdestimate} and \eqref{eq:entropyU:lowrank} that:
\begin{equation}
   H_\epsU(U')\leq (2s+2f+1)2\kappa\log(9/\epsU)+2(s+f)\log\frac{en}{2\min(s,f)}.
   \label{eq:entropy:slA}
\end{equation}
The \emph{bilinear and sparse model} is here the special case for $\kappa=1$
($\slmat_{s,f}=\slmat^1_{s,f}$ in \eqref{eq:slmat})
and, once $\epsU$ does not depend on $n$, entropy scales at most as $\Order((s+f)\log n)$ for sufficiently large $n$.
Again, $U$ has here the property \propA{\eucmod} for $\eucmod=\sqrt{2}$.

\mysubsubsec{Sparse Bilinear Case with one Known Input}
Lemma \ref{lemma:intrinsic:approx} and Lemma \ref{lemma:domainapprox} do not require
that $V$ is a \emph{full} difference set. Here, we essentially consider the set:
\begin{equation}
   V=\bigcup_{\Inx\in\Sigma_s}(B(\Inx\otimes\Sigma_f)-B(\Inx\otimes\Sigma_f))=B(\slmat_{s,2f}).
   \label{eq:model:bilinear:single}
\end{equation}
This case will be relevant when we, universally, have to sample and store measurements 
in a repetitive blind manner whereby  we will have knowledge about one of the components during decoding,
i.e. this comprise a \emph{universal sampling method}.
Thus, once \eqref{eq:approx:rnmp} holds for this rank--one set $U$ 
with $(\alpha,\beta)$ being independent of the ambient dimension
its entropy bound scales additive in
$s$ and $f$, i.e., $\Order((s+f)\log n)$ according to \eqref{eq:entropy:slA} instead of 
$\Order(s\cdot f\log n)$. In our first covering estimate on this set 
in \cite{WJ12b} we have established this scaling for cones directly, not using \cite[Lemma 3.1]{Candes09:LMR}.

\mysubsubsec{The Quadratic and Symmetric Case}
Here, we consider again differences of the form $V=Z-Z$ for $Z=\bigcup_{x\in\Sigma_s} B(x\otimes x)$.
If $B$ is symmetric the binomial formula asserts that:
\begin{equation}
   V=\bigcup_{x,y\in\Sigma_s} B((x+y)\otimes(x-y))=B(\slmat_{2s,2s})
   \label{eq:model:quadratic}
\end{equation}
This model is important for sparse convolutions and sparse phase retrieval as discussed in Section \ref{sec:conv:rnmp}.
Once again, if \eqref{eq:approx:rnmp} holds for $U=\slmat_{2s,2s}$ independent of the ambient dimension,  entropy scales 
linearly in the sparsity $s$, i.e. $\Order(s\log n)$   as follows from 
\eqref{eq:entropy:slA} and \emph{not} as $\Order(s^2\log n)$.\\[.3em]

\subsection{Random Sampling Methods}
Based on the properties \propA{\eucmod} %\eqref{eq:approx:decomposition} 
and \propB{\alpha,\beta} %\eqref{eq:approx:rnmp} 
we consider now random linear mappings $\Phi:V\rightarrow W$ 
where for a small $\delta<1$ the condition $|\lVert \Phi v\rVert-\lVert v\rVert|\leq\delta\lVert v\rVert$ 
should hold simultaneously for all $v\in V=B(U)$ with high probability. For difference sets $V=Z-Z$ (meaning that $U=M-M$ 
for another set $M$ since $B$ is linear)
this condition provides a stable embedding of $Z$ in $W$ and, by \eqref{eq:approx:rnmp},
\emph{it always implies stable embedding $M$ in $W$ -- but
  in the anisotropic situation}.
An estimate for the RIP--like constant $\hat\delta$ of the composed map $\Phi B:U\rightarrow W$
follows with $\alpha=(1-\eta)\xi$ and $\beta=(1+\eta)\xi$ as:
\begin{equation}
   \begin{split}
      |\lVert \Phi B(u)\rVert-\xi\lVert u\rVert| 
      &\leq |\lVert \Phi B(u)\rVert-\lVert B(u)\rVert| + |\lVert B(u)\rVert-\xi\lVert u\rVert| \\
      &\leq \delta\lVert B(u)\rVert + \eta\xi\lVert u\rVert\leq ((1+\eta)\delta\xi+\eta\xi)\lVert u\rVert\\
      &=\xi((1+\eta)\delta+\eta)\lVert u\rVert
      =\xi(\delta+\eta(\delta+1))\lVert u\rVert=:\xi\hat\delta\lVert u\rVert\\
    \end{split}
\end{equation}
The term $\eta(\delta+1)$ reflects the degree of anisotropy caused by $B$. 
A similar relation for the usual definition of the RIP--property has been obtained for example in \cite{Rauhut2008:anisotropic}.
Although we not discuss efficient recovery here, recall that for example \cite{Cai2013} states that for $\hat\delta<1/3$ certain convex recovery methods ($\Ell{1}$--minimization for sparse vectors and nuclear norm minimization for low rank matrices when
$\lVert\cdot\rVert$ are Euclidean norms)
are successful, implying $\eta<1/3$.

\mysubsubsec{Random Model with Generic Concentration}
As shown already in the sparse vector case in \cite[Thm 5.2]{Baraniuk2008} we 
have in this generalized setting a similar statement:
\begin{lemma}   
   Let $\Phi: V\rightarrow W$ be a random linear map which obeys for $\del\in (0,1),\gam>0$ the uniform bound
   $\Pr(\{|\lVert\Phi r\rVert-\lVert r\rVert|\leq\frac{\delta}{2}\lVert r\Vert\})\geq 1-e^{-\gamma}$ for each $r\in V$.
   Let $B:U\rightarrow V$ linear with property \propB{\alpha,\beta}
   where $U$ is a linear cone having property \propA{\eucmod}. Then:
   \begin{equation}
      \Pr(\{\forall v\in V\,:\,|\lVert\Phi v\rVert-\lVert v\rVert|\leq\delta\lVert v\rVert\})\geq 1-e^{-(\gamma-H_\epsU(U'))}
      \label{eq:lemma:randomsampling}
   \end{equation}
   where $\epsU<\frac{\alpha}{7\beta\eucmod}\delta$.
   \label{lemma:randomsampling}
\end{lemma}
\begin{proof}
   From \eqref{eq:lemma:randomsampling} it follows that it is sufficient to consider the set 
   $V'=\{v/\lVert v\rVert\,:\, 0\neq v\in V\}$. 
   From Corollary \ref{cor:UVcovering} we have for this set a covering $\epsV$--net $R$ with respect to an
   intrinsic distance
   of cardinality $|R|\leq e^{H_{\epsV}(V')}\leq e^{H_{\epsU}(U')}$ with $\epsU=\alpha\epsV/(\beta\eucmod)$. 
   Taking the union bound over $R$ asserts therefore that
   $|\lVert\Phi r\rVert-\lVert r\rVert|\leq\frac{\delta}{2}\lVert r\rVert$ with probability $\geq 1-e^{-(\gamma-H_\epsU(U'))}$ 
   for all $r\in R$ and the same $\Phi$. From  Lemma \ref{lemma:intrinsic:approx}, 
   if $\epsV=\frac{\beta\eucmod}{\alpha}\epsU<\delta/7$ there holds
   $|\lVert\Phi v\rVert-\lVert v\rVert|\leq\delta\lVert v\rVert$ for all $v\in V$ and the same $\Phi$ simultaneously
   with probability exceeding $1-e^{-(\gamma-H_\epsU(U'))}$.
   \qed
\end{proof}
This lemma shows that the concentration exponent $\gamma$ must be in the order of the entropy $H_\epsU(U')$ to ensure embedding with 
sufficiently high probability. By construction such a random embedding is a \emph{universal sampling method} where
the success probability in \eqref{eq:lemma:randomsampling} depends solely on the entropy and not on the
particular ''orientation'' of $U'$ which has several practical--relevant advantages as discussed already in the introduction.

\mysubsubsec{Randomizing Fixed RIP Matrices}
We extent the statement of Lemma \ref{lemma:randomsampling} to include randomized classical RIP matrices,
i.e. $\Phi$ is $(k,\delta_k)$--RIP if $|\lVert\Phi v\rVert_2-\lVert v\rVert_2|\leq\delta_k\lVert v\rVert_2$
for each $k$--sparse vector $v$.
The motivation behind is the use of structured or deterministic 
measurements with possibly fast and efficient transform implementation. 
Such measurements usually \emph{fail to be universal} and do not have concentration properties.
However, the important result of \cite{Krahmer2011} states that this can be achieved by a moderate amount
of randomization. Randomization can for example be done with a multiplier $D_\xi$ performing point--wise 
multiplication with a vector $\xi$ having i.i.d. $\pm1$ components, see here also
\cite{Krahmer2012} for more general $\xi$. We consider now $V\subseteq\mathbb{C}^n$ and
$\ell_2$--norms.
\begin{lemma}
   Let $B:U\rightarrow V$ and $U$ as in Lemma \ref{lemma:randomsampling}
   and the random matrix $D_\xi$ is distributed as given above.
   Let $\delta,\rho>0$ and $\Phi$ be $(k,\delta_k)$--RIP with $\delta_k\leq\delta/8$ and
   $k\geq 40(\rho+H_\epsU(U')+3\log(2))$.  
   Then
   \begin{equation}
      \Pr(\{\forall v\in V\,:\,|\lVert\Phi D_\xi v\rVert_2-\lVert v\rVert_2|\leq\delta\lVert v\rVert_2\})\geq 1-e^{-\rho}
   \end{equation}
   where $\epsU<\frac{\alpha}{7\beta\eucmod}\delta$.
   \label{lemma:fixedrip}
\end{lemma}
\begin{proof}
   For a given $(k,\delta_k)$--RIP matrix $\Phi$ with
   $k\geq 40(\rho+p+\log(4))$ and $\delta_k\leq\frac{\delta}{8}$ it follows from \cite{Krahmer2011}:
   $\Phi D_\xi$ is
   with probability $\geq 1-e^{-\rho}$ a $\frac{\delta}{2}$--Johnson--Lindenstrauss--embedding for any 
   point cloud of cardinality $e^p$.
   Now, from Corollary \ref{cor:UVcovering}, there exists an $\epsV$--net $R$ for $V'$ of
   cardinality $|R|\leq e^{H_\epsV}$ where $H_\epsV=H_\epsV(V')\leq H_\epsU(U')$ 
   with $\epsU=\alpha\epsV/(\beta\eucmod)$. When adding the zero--element to the point cloud
   it has cardinality:
   \begin{equation}
      |R|\leq e^{H_\epsV}+1=e^{H_\epsV}(1+e^{-H_\epsV})\leq 2e^{H_\epsV}=e^{H_\epsV+\log(2)}
   \end{equation}
   Therefore, set $p=H_\epsV+\log(2)$ (or the next integer). From \cite{Krahmer2011} it follows then that
   for each $k\geq40(\rho+H_\epsV+\log(2)+\log(4))=40(\rho+H_\epsV+3\log(2))$ 
   the point cloud $R$ is mapped almost--isometrically (including norms since $0$ is included), i.e.
   with probability $\geq 1-e^{-\rho}$ we have 
   $|\lVert\Phi D_\xi r\rVert_2^2-\lVert r\rVert_2^2|\leq\frac{\delta}{2}\lVert r\rVert^2_2$ 
   for all $r\in R$ which implies:
   \begin{equation}
      \Pr(\{\forall r\in R\,:\,|\lVert\Phi D_\xi r\rVert_2-\lVert r\rVert_2|\leq\frac{\delta}{2}\lVert r\rVert_2\})\geq 
      1-e^{-\rho}.
   \end{equation}
   We will choose $\epsV=\frac{\beta\eucmod}{\alpha}\epsU<\delta/7$. Then,
   since $R$ is an $\epsV$--net for $V'$ and $U$ has property \propA{\eucmod} inducing
   an intrinsic decomposition and distance,
   it follows from Lemma \ref{lemma:domainapprox} that: 
   \begin{equation}
      \Pr(\{\forall v\in V\,:\,|\lVert\Phi D_\xi v\rVert_2-\lVert v\rVert_2|\leq\delta\lVert v\rVert_2\})
      \geq 1-e^{-\rho}\quad\qed
   \end{equation}
\end{proof}

\mysubsubsec{Randomizing Random RIP Matrices}
We extent Lemma \ref{lemma:fixedrip} to random structured RIP models which itself are in many cases not universal and
can therefore without further randomization not be used directly in the generalized framework.
Assume an ''$(M,p)$ RIP model'', meaning
that the $m\times n$ random matrix $\Phi$ is 
$(k,\delta_k)$--RIP with probability $\geq 1-e^{-\gamma}$ and $\delta_k\leq\delta$ if $m\geq c\delta^{-2} k^p M(n,k,\gamma)$
for a constant $c>0$. Define for a given $U$:
\begin{equation}
   k_\epsU(\rho):=40(\rho+H_\epsU(U')+3\log(2))
\end{equation}
We have the following lemma:
\begin{lemma}
   Let $\delta>0$ and $D_\xi$, $B:U\rightarrow V$ and $U$ as in Lemma \ref{lemma:fixedrip}.
   Let $\Phi$ be an $m\times n$ random $(M,p)$--RIP model (independent of $D_\xi$) and 
   $k_\epsU(\rho)$ as given above for $\epsU<\frac{\alpha}{7\beta\eucmod}\delta$.
   Then $\Phi D_\xi$ is universal in the sense that:
   \begin{equation}
      \Pr(\{\forall v\in V\,:\,|\lVert\Phi D_\xi v\rVert_2-\lVert v\rVert_2|\leq\delta\lVert v\rVert_2\})\geq 
      1-(e^{-\rho}+e^{-\gamma})
      \label{eq:lemma:randomrip}
   \end{equation}
   if $m\geq 64c\delta^{-2} k_\epsU(\rho)^p M(n,k_\epsU(\rho),\gamma)$.
   \label{lemma:randomrip}
\end{lemma}
\begin{proof}
   The proof follows directly from Lemma \ref{lemma:fixedrip}. Define $\delta'=\delta/8$. Then the model assumptions assert that
   for $m\geq c\delta'^{-2} k_\epsU(\rho)^p M(n,k_\epsU(\rho),\gamma)$ the matrix $\Phi$ has $(k_\epsU(\rho),\delta_k)$--RIP
   with $\delta_k\leq\delta'=\delta/8$ with probability $\geq 1-e^{-\gamma}$. Thus, by Lemma  \ref{lemma:fixedrip}
   for any $\rho>0$ the claim follows. \qed
\end{proof}
The best $(\rho,\gamma)$--combination for a fixed probability bound $\geq 1-e^{-\lambda}$ can be estimated by 
minimizing $k_\epsU(\rho)^p M(n,k_\epsU(\rho),\gamma)$.
We will sketch this for random \emph{partial circulant matrices} $P_\Omega\hat D_\eta$. 
Let $F=\left(e^{i2\pi kl/n}\right)_{k,l=0}^{n-1}$ be the $n\times n$--matrix of the 
(non--unitary) discrete Fourier transform. Then, 
$\hat D_\eta:=F^{-1} D_\eta F$ 
is an $n\times n$ circulant matrix with $\hat\eta:=F\eta$ on its first row (Fourier multiplier $\eta$) and
the $m\times n$ matrix $P_\Omega:=\frac{1}{|\Omega|}1_{\Omega}$ is the normalized 
projection onto coordinates in the set $\Omega\subset[1,\dots n]$ of size $m=|\Omega|$.
Random convolutions for compressed sensing are already proposed in \cite{Romberg2009}. In  \cite{Tropp:RandDemod} a related 
approach has been called \emph{random demodulator} and is used for sampling frequency--sparse
signals via convolutions on the Fourier side (being not suitable for sporadic communication tasks). 
Measurement matrices $P_\Omega\hat D_\eta$ are systematically investigated in \cite{Rauhut2012} showing that
$(k,\delta_k)$--RIP properties hold in the regime $m=\Order((k\log n)^\frac{3}{2})$. Finally, linear scaling in $k$ 
(and this will be necessary for the overall additivity statement in the bilinear setting) has been 
achieved in \cite{Krahmer2012}. But $P_\Omega\hat D_\eta$ is \emph{not universal} meaning that
the signal has to be $k$--sparse in the canonical basis. 

Therefore, we propose the \emph{universal random demodulator} $P_\Omega\hat D_\eta D_\xi$
which still has an efficient FFT--based implementation but is independent of the sparsity domain. 
Such random matrices work again in our framework:
\begin{lemma}
   Let be $D_\xi$, $B:U\rightarrow V$ and $U$ as in Lemma \ref{lemma:fixedrip}.
   If $\Phi=P_\Omega\hat D_\eta$ is a $m\times n$ partial random circulant matrix with $\eta$ 
   being i.i.d. zero--mean, unit--variance and subgaussian vector
   with:
   \begin{equation}
      m\geq 64c\delta^{-2}(\lambda+h_\epsU)\max((\log(\lambda+h_\epsU)\log(n))^2,\lambda+\log(2))
   \end{equation}
   where $h_\epsU=H_\epsU(U)+4\log(2)$. If $\epsU<\frac{\alpha}{7\beta\eucmod}\delta$
   the LHS of statement \eqref{eq:lemma:randomrip} holds with probability
   $\geq 1-e^{-\lambda}$.
   \label{lemma:random:partialcirculant}
\end{lemma}
\begin{proof}
   From \cite[Theorem 4.1]{Krahmer2012} we have that:
   \begin{equation}
      M(n,k,\gamma)=\max((\log(k)\log(n))^2,\gamma)
   \end{equation} and $p=1$ in Lemma \ref{lemma:randomrip}.
   We choose $\rho=\gamma=:\lambda+\log(2)$ (being suboptimal). \qed 
\end{proof}
Since this choice $\rho$ and $\gamma$ is not necessarily optimal the logarithmic order in $n$ might be improved. However, 
for fixed $\lambda$ and sufficiently small $\epsU$ we have
$m=\Order(h_\epsU[\log(h_\epsU)\log(n)]^2)$ which is sufficient to preserve, for 
example, additive scaling (up to logarithms and large $n$) for the bilinear sparse models 
once $\epsU$ does not depend on $n$ and where $h_\epsU=\Order((s+f)\log n)$.

\mysubsubsec{Stable Embedding of Bilinear Signal Sets}
Finally, we come back now to the application for bilinear inverse problems with sparsity priors as discussed in the
introduction.  From the communication theoretic and signal processing point of view we will consider the problems (i)
and (ii) on page \pageref{item:task:1} and we give the results for both cases in one theorem.
Although we will summarize this for generic random measurements due to concentration as in Lemma \ref{lemma:randomsampling},
it follows from Lemma \ref{lemma:random:partialcirculant} that \emph{the scaling even remains 
  valid in a considerable de--randomized setting}.  
The assertion $(i)$ in the next theorem was already given in \cite{WJ12b}.
Recall, that $\slmat_{s,f}\subseteq \Complex^{n\times n}$ are
the $(s,f)$--sparse rank--one matrices as defined
in \eqref{eq:slmat}.
\begin{theorem}
   Set (i) $U=\slmat_{s,f}$ and $\kappa=1$ or (ii) $U=\slmat_{s,f}-\slmat_{s,f}$ and $\kappa=2$ equipped with
   the Frobenius norm.
   Let be $B:U\rightarrow V\subseteq\Complex^n$ linear 
   with property \propB{\alpha,\beta} 
   and $\lVert\cdot\rVert$ be a norm in $V$. If
   $\alpha,\beta$ not depend on $n$, 
   $\Phi\in\Complex^{m\times n}$ obeys
   $\Pr(\{|\lVert\Phi r\rVert-\lVert r\rVert|\leq\frac{\delta}{2}\lVert r\Vert\})\geq 1-e^{-c\delta^2m}$ 
   for each $r\in V$ 
   and $m\geq c''\delta^{-2}(s+f)\log(n/(\kappa\min(s,f)))$ it follows that:
   \begin{equation}
      \Pr(\{\forall v\in V\,:\,|\lVert\Phi v\rVert-\lVert v\rVert|\leq\delta\lVert v\rVert\})\geq 1-e^{-c'm}
   \end{equation}
   were  $c',c''>0$ (only depending on $\del$).
   \label{lemma:randomsampling:bilinear}
\end{theorem}
\begin{proof}
   In both cases $U$ has property \propA{\eucmod} with $\eucmod=\sqrt{2}$.
   Fix exemplary $\epsU:=\frac{\alpha}{8\beta\sigma}\delta<\frac{\alpha}{7\beta\sigma}\delta$ for Lemma \ref{lemma:randomsampling}.
   From \eqref{eq:entropy:slA} we have in both cases (i) and (ii):
   \begin{equation}
      \begin{split}
         H_\epsU(U)
         &\leq(s+f+\frac{1}{2})4\kappa\log(9/\epsU)+\kappa(s+f)\log(n/(\kappa\min(s,f)))\\
         &=(s+f+\frac{1}{2})4\kappa\log\frac{8\cdot9\eucmod\beta}{\alpha\delta}+\kappa(s+f)\log\frac{n}{\kappa\min(s,f)}=:h_\delta\\
      \end{split}
   \end{equation}
   where $(\alpha,\beta)$ are the bounds for $B$ in \eqref{eq:approx:rnmp} and independent of $n$. Let
   $\gamma=c\delta^2m$ for some $c>0$. We have from Lemma 
   \ref{lemma:randomsampling}:
   \begin{equation}
      \begin{split}
         \Pr(\{\forall v\in B(U)\,:\,|\lVert\Phi v\rVert-\lVert v\rVert|\leq\delta\lVert v\rVert\})
         &\geq 1-e^{-(c\delta^2 m-h_\delta)}.
      \end{split}
   \end{equation}
   To achieve exponential probability of the form $\geq 1-\exp(-c' m)$ we have to ensure a constant $c'>0$ such that
   $c\delta^2 m-h_\delta\geq c'm$. In other words $\delta^2(c-\frac{\delta^{-2}h_\delta}{m})\geq c'>0$
   meaning that must be a constant $c''$ such that the number of measurements fulfill
   $m\geq c''\delta^{-2}(s+f)\log(n/(\kappa\min(s,f)))$. \qed
\end{proof}

\mysubsubsec{Final Remarks on Recovery}
In this section we have solely discussed embeddings. Hence, it is not at all clear that one can achieve
recovery in the $s+f$--regime even at moderate complexity. A negative result has been shown here already in
\cite{Oymak2012} for multi--objective convex programs which are restricted to the 
Pareto boundary caused by the individual objectives. On the other hand, greedy algorithms or alternating minimization algorithms
like the ''sparse power factorization'' method \cite{Lee:SPF} seems to be capable to operate in the desired
regime once the algorithm is optimally initialized.

\section{Sparse Convolutions and Stability}
\label{sec:conv:rnmp}

In this section, we will consider the central condition \eqref{eq:approx:rnmp} for
the special case where the bilinear mapping $B$ refers to convolutions representing for example
basic single--dispersive communication channels.
Let us start with the case where $B(\Inx,\Iny)=\Inx\circledast\Iny$ is
given as the \emph{circular} convolution in $\C^n$. Denote with $k\ominus i$ the difference
$k-i$ modulo $n$. Then this bilinear mapping is defined as:
\begin{align}
  (\vx \circledast \vy)_{k}= \sum_{i=0}^{n-1} x_i y_{k\ominus i} \foral k \in \{0,\dots,n-1\} \label{eq:defcconv}.
\end{align}
Our analysis was originally motivated by the work in \cite{Hegde2011} where the authors 
considered circular convolutions with $x\in\Sigma_s$ and $y\in\Sigma_f$.
We will show that under certain conditions circular convolutions fulfill
property \eqref{eq:approx:rnmp} for $U=\{\Inx\otimes\Iny\,:\, \Inx\in X,\Iny\in Y\}$ and suitable sets
$X,Y\subset \C^{2n-1}$. In this case \eqref{eq:approx:rnmp} reads as:
\begin{equation}   
   \alpha\lVert\Inx\rVert\lVert\Iny\rVert\leq \lVert \Inx \circledast \Iny \rVert\leq\beta\lVert\Inx\rVert\lVert\Iny\rVert
   \quad\text{for all}\quad (\Inx,\Iny)\in X\times Y, 
   \label{eq:approx:rnmp:1}
\end{equation}
where from now on $\lVert x\rVert:=\lVert x\rVert_2$ will always denotes the $\ell_2$--norm.  According to \cite{WJ12b}
we call this condition \emph{restricted norm multiplicativity property} (RNMP).  As already pointed out in the previous
section, this condition ensures compression for the models \eqref{eq:model:bilinear:single} and
\eqref{eq:model:quadratic} as summarized in Theorem \ref{lemma:randomsampling:bilinear}.
In fact, \eqref{eq:approx:rnmp:1} follows as a special case of sparse convolutions, if one restrict the support of $x$
and $y$ to the first $n$ entries. In this case circular convolution \eqref{eq:defcconv} equals (ordinary) convolution
which is defined on $\Z$ element-wise for absolute--summable $x,y \in \Ell{1}(\Z)$ by
\begin{align}
  (\vx * \vy)_{k}= \sum_{i\in \Z} x_i y_{k-i} \foral k \in \Z \label{eq:defconv}.
\end{align}
Obviously, the famous Young inequality states for $1/p+1/q-1/r=1$ and $1\leq p,q,r\leq\infty$ that:
\begin{equation}
   \Norm{\Inx \ast \Iny}_r \leq \Norm{\vx}_p \Norm{\vy}_q
   \label{eq:young}
\end{equation}
and implies sub-multiplicativity of convolutions in $\ell_1$
but a reverse inequality was only known for positive signals. 
However, the same is true when considering \emph{$s$--sparse sequences}
$\Sigma_s=\Sigma_s(\Z)$ as we will show in \thmref{thm:dryi}. Moreover,
the lower bound $\alp$ in \eqref{eq:approx:rnmp:1} depends solely on the sparsity levels of the signals and not on the support location.  

Let us define $[n]:=\{0,\dots,n-1\}$ and the set of subsets with cardinality $s$ by $[n]_s:=\set{T\subset [n]}{|T|=s}$.
For any $T\in [n]_s$ the matrix $\vB_T$ denotes the $s\times s$ principal submatrix of $\vB$ with rows and columns in $T$.  Further, we denote by
$\vB_\va$ an $n\times n-$ \emph{Hermitian Toeplitz matrix}  generated by $\va\in \Sigma^n_s$ with \emph{symbol} given
for $\ome\in [0,2\pi)$ by 
% %
% \begin{equation}
%    b (\va,\ome)=  \sum_{k=-n+1}^{n-1} b_k(\va) e^{\imath k\ome}
%    = 1 + \sum_{k=1}^{n-1}\left( \mu_k \cos(k\ome) + \nu_k \sin(k\ome)\right)
%    \label{eq:bsymbol3}
% \end{equation}
% where $\mu_k:=2\Re(b_k(\va))$ and $\nu_k:=-2\Im(b_k(\va))$ and
\begin{equation}
   b (\va,\ome)=  \sum_{k=-n+1}^{n-1} b_k(t) e^{\imath k\ome},
   \label{eq:bsymbol3}
\end{equation}
which for $b_k(t):=(t*\cc{t^-})_k$ defines a \emph{positive trigonometric polynomial of order not larger than $n$} by
the \namen{Fejer-Riesz} factorization. Note, $b_k$ are the samples of the auto-correlation of $t$ which can be written
as the convolution of $t=\{t_k\}_{k\in\Z}$ with the complex-conjugation of the time reversal $t^{-}$,
given component-wise by $t^{-}_k=t_{-k}$.  We will need a notion of the $k$--restricted determinant:
\begin{equation}
   D_{n,k}:=\min \{|\det(B_t)|\,:\,t\in\Sigma^n_k\,,\,\lVert t\rVert=1\}
   \label{eq:mindet}
\end{equation}
which exists by compactness arguments.
%
%For $s\leq f\leq n$, that all $(s,f)-$sparse circular convolutions embedded in $\ko^{\pn}$ with $\pn:=2n-1$ by adding
%$n-1$ zeros, have a universal $\ell^2$-norm lower bound.  The norm of the circular convolution is then equal to the
%discrete convolution. 

\subsection{The RNMP for Sparse Convolutions}
\label{sec:rnmp:sparse}

The following theorem is a generalization of a result in \cite{WJ13a}, (i) in the sense of the extension to infinite
sequences on  $\Z$ (ii) extension to the complex case, which actually only replaces \namen{Szeg{\"o}} factorization with
\namen{Fejer-Riesz} factorization in the proof and (iii) with a precise determination of the dimension parameter
$n$\footnote{Actually, the estimate of the dimension $n=\tn$ of the constant $\alp_{\tilde{n}}$ in \cite{WJ13a}, was
quite too optimistic.}.
%Moreover, we already have embedded the finite dimensional signals in $\ell^2$ and instead formulated the result for the
%discrete convolution of signals with finite support\footnote{Actually, the estimate of the dimension $\tn$ of the
%constant $\alp_{\tn}$ in \cite{WJ12b}, was quite too optimistic.}.
%
%%%%%%%%%%%%%%%%%%%%%%%%%%%%%%%%%%%%%%%%%%%%%%%%%%%%%%%%%%%%%%%%%%%%%%%%%%%%%%%%%%%%%%%%%%%%%%%%%%%%%%%%%%%%%%%%%
%%%%%%%%%%%%%%%%% MAIN THEOREM %%%%%%%%%%%%%%%%%%%%%%%%%%%%%%%%%%%%%%%%%%%%%%%%%%%%%%%%%%%%%%%%%%%%%%%%%%%%%%%%%%
%%%%%%%%%%%%%%%%%%%%%%%%%%%%%%%%%%%%%%%%%%%%%%%%%%%%%%%%%%%%%%%%%%%%%%%%%%%%%%%%%%%%%%%%%%%%%%%%%%%%%%%%%%%%%%%%%
\begin{theorem}\label{thm:dryi}
   For $s,f\in\N$ exist constants $0<\alp(s,f)\leq\beta(s,f)<\infty$ 
   such that for all $\vx\in\Sigma_s$ and $\vy\in \Sigma_f$ it holds:
   \begin{align}
      \alp(s,f) \Norm{\vx}\Norm{\vy} \leq \Norm{\vx * \vy} \leq \beta(s,f)\Norm{\vx}\Norm{\vy} \label{eq:dryi},
   \end{align}
   where $\bet^2(s,f)=\min\{s,f\}$. Moreover, the lower bound  only depends on the sparsity levels $s$ and $f$
   of the sequences and can be lower bounded by
   \begin{align}
     \alp^2 (s,f)\geq \frac{1}{\sqrt{n\cdot\min(s,f)^{n-1}}}\cdot D_{n,\min(s,f)}      
      \label{eq:lowerbound},
   \end{align}
   with $n=\lfloor 2^{2(s+f-2)\log(s+f-2)}\rfloor$. This bound is decreasing in $s$ and $f$. For $\beta(s,f)=1$ it
   follows that $\alp(s,f)=1$. 
\end{theorem}
%%%%%%%%%%%%%%%%%%%%%%%%%%%%%%%%%%%%%%%%%%%%%%%%%%%%%%%%%%%%%%%%%%%%%%%%%%%%%%%%%%%%%%%%%%%%%%%%%%%%%%%%%%%%%%%%%
%%%%%%%%%%%%%%%%%%%%%%%%%%%%%%%%%%%%%%%%%%%%%%%%%%%%%%%%%%%%%%%%%%%%%%%%%%%%%%%%%%%%%%%%%%%%%%%%%%%%%%%%%%%%%%%%%
%%%%%%%%%%%%%%%%%%%%%%%%%%%%%%%%%%%%%%%%%%%%%%%%%%%%%%%%%%%%%%%%%%%%%%%%%%%%%%%%%%%%%%%%%%%%%%%%%%%%%%%%%%%%%%%%%
%
The main assertion of the theorem is: The smallest $\ell^2-$norm over all convolutions of $s-$ and $f-$sparse normalized
sequences can be determined solely in terms of $s$ and $f$, where we used the fact that the sparse convolution can be
represented by sparse vectors in $n=\lfloor 2^{2(s+f-2)\log(s+f-2)}\rfloor$ dimensions, due to an additive combinatoric
result.  An analytic lower bound for $\alp$, which decays exponentially in the sparsity, has been found very recently in
\cite{WJP14}.  Although $D_{n,\min\{s,f\}}$ is decreasing in $n$ (since we extend the minimum to a larger set by
increasing $n$) nothing seems to be known on the precise scaling in $n$.  Nevertheless, since $n$ depends solely on $s$
and $f$ it is sufficient to ensure that $D_{n,\min\{s,f\}}$ is non--zero. 
\begin{proof}
   The upper bound is trivial and follows, for example, from the Young inequality \eqref{eq:young} for $r=q=2$ and $p=1$
   and with the Cauchy--Schwartz inequality, i.e., in the case $s\leq f$ this yields:
   \begin{equation}
      \lVert x *y\rVert_2\leq\lVert x\rVert_1\lVert y\rVert_2\leq\sqrt{s}\lVert x\rVert_2\lVert y\rVert_2
      \label{eq:upperbound:young}.
   \end{equation}   
   For $\vx=0$ or $\vy=0$ the inequality is trivial as well, hence we assume that $\vx$ and $\vy$ are non-zero.  We
   consider therefore the following problem: 
   \begin{equation}
     \inf_{\substack{(\vx,\vy)\in(\Sigma_s,\Sigma_f)\\\vx\not=0\not=\vy}} \frac{\Norm{\vx * \vy}}{\Norm{\vx}\Norm{\vy}}=
      \inf_{\substack{(\vx,\vy)\in(\Sigma_s,\Sigma_f)\\ \Norm{\vx}=\Norm{\vy}=1}} \Norm{\vx * \vy}
      \label{eq:pb}.
   \end{equation}
 Such \emph{bi-quadratic optimization problems} are known to be NP-hard in general \cite{LNQY09}.
 According \eqref{eq:defconv} the squared norm can be written as:
 \begin{align}
   \Norm{\vx*\vy}^2=\sum_{k\in \Z} \left|\sum_{i\in \Z} x_{i} y_{k-i}\right|^2.\label{eq:normsquarestart}
 \end{align}
 Take sets $I,J\subset \Z$ such that $\supp(\vx)\subseteq I$ and $\supp(\vy)\subseteq J$ with $|I|= s,|J|= f$ and
 let $I=\{i_0,\dots,i_{s-1}\}$ and $J=\{j_0,\dots,j_{f-1}\}$ (ordered sets).  
 Thus, we represent $\vx$ and $\vy$ by complex vectors $\hvx \in \C^s$ and $\hvy \in \C^f$ component--wise, i.e.,
 for all $i,j\in \Z$:
 \begin{align}
   x_i = \sum_{\tht=0}^{s-1} \hx_\tht \del_{i,i_\tht}\quad\text{and}\quad  y_j = \sum_{\gam=0}^{f-1} \hy_\gam
   \del_{j,j_\gam}.   
 \end{align}
 Inserting this representation in \eqref{eq:normsquarestart} yields: 
 \begin{align}
   \Norm{\vx*\vy}^2&=\sum_{k\in \Z}\Bigg|\sum_{i\in \Z}\left( \sum_{\tht=0}^{s-1} \hx_{\tht}\del_{i,i_{\tht}}\right)
   \left( \sum_{\gam=0}^{f-1} \hy_\gam \del_{k-i,j_\gam}\right)\Bigg|^2 \label{eq:umalv}\\
   &= \sum_{k\in \Z} \Bigg| \sum_{\tht=0}^{s-1} \sum_{\gam=0}^{f-1}  \sum_{i\in \Z} \hx_{\tht}\del_{i,i_{\tht}} 
     \hy_\gam \del_{k,j_\gam +i}\Bigg|^2.\label{eq:ijsum}
\intertext{Since the inner $i-$sum is over  $\Z$, we can shift $I$ by $i_0$ if we set $i\ra i+i_0$ (note that $\vx\not=0$), 
     without changing the value of the sum: }
   &= \sum_{k\in \Z} \Bigg| \sum_\tht \sum_\gam 
   \sum_{i\in \Z} \hx_{\tht} \del_{i+{i_0},i_{\tht}} \hy_\gam \del_{k,j_\gam +i + i_0} \Bigg|^2.
   \intertext{By the same argument we can shift $J$ by $j_0$ by setting $k\ra k +i_0 +j_0$ and get:}
   &= \sum_{k\in \Z} \Bigg| \sum_\tht \sum_\gam 
    \sum_{i\in \Z} \hx_{\tht} \del_{i,i_{\tht}-i_0} \hy_\gam \del_{k,j_\gam -j_0 + i} \Bigg|^2.
   \intertext{Therefore we always can assume that the supports $I,J\subset \Z$ fulfill $i_0=j_0=0$ in $\eqref{eq:pb}$. 
              From \eqref{eq:ijsum} we get:}
   &= \sum_{k\in \Z} \Bigg| \sum_\tht \sum_\gam  \hx_{\tht} \hy_\gam \del_{k,j_\gam +{i_\tht}} \Bigg|^2 \\
   &= \sum_{k\in \Z} \sum_{\tht,\tht'}\sum_{\gam,\gam'}  \hx_{\tht} \cc{\hx_{\tht'}} \hy_\gam \cc{\hy_{\gam'}} 
   \del_{k,j_{\gam} +i_{\tht}} \del_{k,j_{\gam'} +i_{\tht'}}\\
   &=\sum_{\tht,\tht'} \sum_{\gam,\gam'} 
   \hx_\tht \cc{\hx_{\tht'}} \hy_{\gam} \cc{\hy_{\gam'}} \del_{i_{\tht} +j_\gam,j_{\gam'} +i_{\tht'}}
    \label{eq:pipjhxhy}.
 \end{align}
 The interesting question is now: \emph{What is the smallest dimension $n$ to represent this fourth order tensor
   $\del_{{i_\tht}+j_{\gam},i_{\tht'} +j_{\gam'}}$, i.e. representing the additive structure?} Let us consider an
   ''index remapping'' $\phi:A \to \Z$ of the indices $A\subset \Z$.  Such a map $\phi$ which preserves additive
   structure:
 \begin{align}
    a_1+a_2=a'_1 + a'_2\,\Rightarrow\, \phi(a_1) + \phi(a_2) = \phi(a'_1) + \phi(a'_2) \label{eq:f2homodef}
 \end{align}
 for all $a_1,a_2,a'_1,a'_2\in A$
 is called a \emph{Freiman homomorphism} on $A$ of order $2$  and a \emph{Freiman isomorphism} if:
 \begin{align}
    a_1+a_2=a'_1 + a'_2 \LRA \phi(a_1) + \phi(a_2) = \phi(a'_1) + \phi(a'_2) \label{eq:f2isodef}
 \end{align}
 for all $a_1,a_2,a'_1,a'_2\in A$, see e.g. \cite{Tao06,Gry13}. For $A:=I\cup J$ the property \eqref{eq:f2isodef} gives exactly our desired indices
 $\phi(I)$ and $\phi(J)$ and we have to determine $n=n(s,f)$ such that $\phi(A)\subset [n]$. 
 The minimization problem reduces then to an $n$--dimensional problem.
 Indeed, this was a conjecture in \cite{KL00} and very recently proved in
 \cite[Theorem 20.10]{Gry13} for sets with Freiman dimension $d=1$. Fortunately, he could prove a more general compression argument for
 arbitrary sum sets in a torsion-free abelian group $G$ having a finite Freiman dimension $d$. We will state here a restricted version of 
 his result for the $1$--dimensional group $G=(\Z,+)$ and $A_1=A_2=A$:
 \begin{lemma}
    \label{thm:gry2}
    Let $A\subset \Z$ be a set containing zero
    with $m:=|A|<\infty$ and Freiman dimension $d=\dim^+(A+A)$. Then there exists an  Freiman isomorphism 
    $\phi: A \to \Z$ of order $2$ such that;
    \begin{align}
      \diam(\phi(A))\leq
     d!^2 \left(\frac{3}{2}\right)^{d-1} 2^{m-2} + \frac{3^{d-1} -1}{2}. 
     \label{eq:diamm}
    \end{align}
 \end{lemma}
 We here use the definition of a Freiman isomorphism according to \cite[p.299]{Gry13} which is a more generalized version as in \cite{Tao06}. In
 fact, $\phi:A\to \Z$ can be easily extended to $\phi':A+A\to \Z$ by setting 
 $\phi'(a_1+a_2)=\phi(a_1)+\phi(a_2)$. Then Grynkiewicz defines the map $\phi'$ to be a Freiman homomorphism, if
 $\phi'(a_1+a_2)=\phi'(a_1) + \phi'(a_2)$ for all $a_1,a_2\in A$. If $\phi'$ is also
 injective, then it holds 
 \begin{align}
    \phi'(a_1) + \phi'(a_2)=\phi'(a'_1) +\phi'(a'_2) \LRA a_1+a_2 = a'_1+a'_2\label{eq:f2isogryn}.
 \end{align}
 Since $0\in A$ we have for every $a\in A$ that $\phi'(a+0)=\phi(a)+\phi(0)$ and therefore \eqref{eq:f2isogryn} is
 equivalent to our definition \eqref{eq:f2isodef}.
 Furthermore, we have
 $\text{diam}(\phi'(A))=\text{diam}(\phi(A)+\phi(0))=\text{diam}(\phi(A))=\max\phi(A)-\min\phi(A)$.\\[.3em]

 We continue with the proof of the theorem by taking $A=I\cup J$. Recall that there always exists sets $I,J\subset G$ with $|I|=s$ and
 $|J|=f$ containing the support of $x$ resp. $y$.
 Since $0\in I\cap J$ we always have $m=|A|\leq s+f-1$.
 Unfortunately, the Freiman dimension can be much larger than the linear dimension of the ambient group $\Z$. But we can
 bound $d$ for any $A\subset \Z$ by a result\footnote{Note, that the Freiman dimension of order $2$ in \cite{Tao06} is
 defined by $\dim(A):=\dim^+(A+A)-1=d-1$.} of Tao and Vu in \cite[Corollary 5.42]{Tao06} by
 \begin{align}
   \min\{|A+A|,|A-A\}\leq \frac{|A|^2}{2} - \frac{|A|}{2} +1 \leq (d+1)|A| -\frac{d(d+1)}{2}
  \end{align}
 where the smallest possible $d$ is given by $d=|A|-2$. Hence we can assume $d\leq m-2$ in \eqref{eq:diamm}. By using
 the bound  $\log (d!)\leq ((d+1)\ln(d+1) - d)/\ln 2$, we get the following upper bound:
 \begin{align}
    \diam(\phi(A)) & < d!^2 \left(\frac{3}{2}\right)^{m-3}\cdot 2^{m-2} + \frac{3^{m-3}}{2}=
    (2 (d!)^2 +2^{-1}){3^{m-3}}\\
    &    <( 2^{[2(m-1)\log(m-1)\ln 2 -2(m -2)]/\ln 2 +1 } +2^{-1}) 3^{m-3}\\
    \intertext{using $3<2^2$ and  $2/\ln 2 >2$ we get}
   &< 2^{2(m-1)\log( m-1) -2(m-2) +1 +2(m-3)} + 2^{2(m-3)-1} \\
   &= 2^{2(m-1)\log( m-1) -1} + 2^{2(m-3)-1}\\
   &<\lfloor2^{2(m-1)\log( m-1) -1} + 2^{2(m-1)-1}\rfloor-1 \\
   &< \lfloor2^{2(m-1)[\log(m-1)-1} +2^{2(m-1)\log(m-1) -1}\rfloor -1 \\
    &=\lfloor2^{2(s+f-2)\log(s+f-2)}\rfloor-1.\label{eq:n2}
  \end{align}
 We translate $\phi$ by $a^*:=\min \phi(A)$, i.e. $\phi'=\phi-a^*$ still satisfying \eqref{eq:f2isodef}. Abbreviate
 $\tI=\phi'(I)$ and $\tJ=\phi'(J)$. From \eqref{eq:diamm} we have with $n=\lfloor 2^{2(s+f-2)\log(s+f-2)}\rfloor$:
 \begin{align}
   0\in\tI\cup \tJ\subset\{0,1,2,\dots ,n-1\}=[n].
 \end{align}
 and by \eqref{eq:f2homodef} for all $\tht,\tht'\in [s]$ and $\gam,\gam'\in [f]$ we have the identity 
 \begin{align}
   \del_{i_{\tht} +j_\gam, i_{\tht'}+ j_{\gam'}} = \del_{\ti_\tht + \tj_\gam,\ti_{\tht'} +
   \tj_{\gam'}}.\label{eq:phiprime}
\end{align}
 Although a Freiman isomorphism does not necessarily preserve the index order this is not important 
 for the norm in \eqref{eq:pipjhxhy}. We define the embedding of $\hvx,\hvy$ into $\C^{n}$ by
 setting for all $i,j \in [n]$:
 \begin{align}
    \tx_i = \sum_{\tht=0}^{s-1} \hx_{\tht} \del_{i,{\ti_{\tht}}} \quad\text{and}\quad
    \ty_j = \sum_{\gam=0}^{f-1} \hy_{\gam} \del_{j,{\tj_{\gam}}} \label{eq:xytn}.
 \end{align}
 Let us further set $\tx_i =\ty_i=0$ for $i\in \Z\setminus [n]$. Then we get  
 from \eqref{eq:pipjhxhy}:
 \begin{align}
   \Norm{x*y}^2&=\sum_{\tht,\tht'} \sum_{\gam,\gam'} 
          \hx_\tht \cc{\hx_{\tht'}} \hy_{\gam} \cc{\hy_{\gam'}} \del_{i_{\tht} +j_\gam,j_{\gam'} +i_{\tht'}}\\
\eqref{eq:phiprime}\ra&=\sum_{\tht,\tht'} \sum_{\gam,\gam'}\hx_\tht \cc{\hx_{\tht'}} \hy_{\gam} \cc{\hy_{\gam'}}
          \del_{\ti_\tht + \tj_\gam,\ti_{\tht'} + \tj_{\gam'}}.\\
\intertext{Going analog backwards as in \eqref{eq:pipjhxhy} to \eqref{eq:umalv} we get}
&=\sum_{k\in \Z}\Bigg|\sum_{i\in \Z}\left( \sum_{\tht=0}^{s-1} \hx_{\tht}\del_{i,\ti_{\tht}}\right)
\left( \sum_{\gam=0}^{f-1} \hy_\gam \del_{k-i,\tj_\gam}\right)\Bigg|^2 \\
\eqref{eq:xytn}\ra &= \sum_{k\in\Z} \Bigg| \sum_{i\in \Z} \tx_{i} \ty_{k-i} \Bigg|^2=\Norm{\tx*\ty}^2.
\end{align}
Furthermore, we can rewrite the Norm by using the support properties of $\tx,\ty$ as
\begin{align}
 \Norm{\tx*\ty}^2  &=\sum_{i,i'=0}^{n-1} \tx_{i}\cc{\tx_{i'}}  \sum_{k\in \Z} \ty_{k-i} \cc{\ty_{k-i'}}\\
 \intertext{and substituting by $k'=k-i$ we get} 
  &=\sum_{i,i'=0}^{n-1}  \tx_{i}\cc{\tx_{i'}} \sum_{k'=\max\{0,i-i'\}}^{\min\{n-1,n-1-(i-i')\}}  \ty_{k'} 
        \cc{\ty_{k' +(i-i')}}=\skprod{\vtx,\vB_{\vty} \vtx}\label{eq:skprodvxvy},
\end{align}
where $\vB_{\vty}$ is an $n\times n$ Hermitian Toeplitz matrix with first row $(\vB_{\vty})_{0,k}=\sum_{j=0}^{n-k}
\cc{\ty_{j}} {\ty_{j+k}}=:b_k(\vty)=(\cc{\ty} * \ty^{-})_k$ resp. first column $(\vB_{\vty})_{k,0}=:b_{-k}(\vty)$ for
$k\in [n]$.  Its  \emph{symbol} $b(\vty,\ome)$ is given by \eqref{eq:bsymbol3} and since $b_0=\Norm{\vty}=1$ it is for
each $\vty\in \C^{n}$ a \emph{normalized  trigonometric polynomial of order $n-1$}. Minimizing the inner product in
\eqref{eq:skprodvxvy} over $\vtx\in \Sigma^{n}_s$ with $\Norm{\vtx}=1$ includes all possible $\tilde{I}$ and therefore
establishes a \emph{lower bound} (see the remarks after the proof). However, this then means to minimize the minimal
eigenvalue $\lambda_{\min}$ over all $s\times s$ principal submatrices of $\vB_{\vty}$:
 \begin{align}
   \lambda_{\min}(\vB_{\vty},s):=\min_{\vtx\in \Sigma^{n}_s, \Norm{\vtx}=1} \skprod{\vtx,\vB_{\vty} \vtx} 
      \geq \lambda_{\min}(\vB_{\vty})
   \label{eq:Bty}
 \end{align}
 whereby $\lambda_{\min}(\vB_{\vty},s)$ is sometimes called the $s$--restricted eigenvalue (singular value) of
 $\vB_{\vty}$, see \cite{rudelson:anisotrop} or \cite{Kueng12}.
 First, we show now that $\lambda_{\min}(\vB_{\vty})>0$.  By the well-known Fejer-Riesz factorization, see e.g.
 \cite[Thm.3]{Dim04}, the symbol of $\vB_{\vby}$  is \emph{non-negative}%
 \footnote{Note, there exist $\vby\in\C^{n}$ with $\Norm{\vby}=1$ and $b(\vby,\ome)=0$ for some $\ome\in[0,2\pi)$. That
 is the reason why things are more complicated here. Moreover, we want to find a universal lower bound over all $\vby$,
 which is equivalent to a universal lower bound over all non-negative trigonometric polynomials of order $n-1$.}
 for every $\vby\in\C^{n}$. By \cite[(10.2)]{BG05a} it follows therefore that 
 \emph{strictly} $\lambda_{\min}(\vB_{\vby})>0$. Obviously, then also the determinant is non--zero.
 Hence
 $\vB_{\vby}$ is invertible and with $\lambda_{\min}(\vB_{\vby}) = 1/\|\vB^{-1}_{\vby}\|$ we can estimate the
 smallest eigenvalue (singular value) by the determinant \cite[Thm. 4.2]{BG05a}:
 \begin{align}
    \lambda_{\min}(\vB_{\vby})\geq |\det(\vB_{\vby})| \frac{1}{\sqrt{n} (\sum_k |b_k(\vby)|^2)^{(n-1)/2}}
    \label{eq:eigdetbound}
 \end{align}
 whereby from $\lVert\vby\rVert=1$ and the upper bound of the theorem or directly \eqref{eq:upperbound:young} it follows
 also that $\sum_k |b_k(\vby)|^2=\|\vby *\cc{\vby^{-}}\|^2\leq f$ if $\vby\in\Sigma^n_f$.  Since the determinant is a
 continuous function in $\vby$ over a compact set, the non--zero minimum is attained. Minimizing \eqref{eq:eigdetbound} over all
 sparse vectors $\vby$ with smallest sparsity yields 
 \begin{equation}
   \min_{\substack{ \vby\in \Sigma^n_{\min\{s,f\}}\\\lVert\vby\rVert=1}  } \lam_{min}(B_{\vby})\geq 
   \sqrt{\frac{1}{n f^{n-1}}} \cdot
   \underbrace{\min_{t\in\Sigma^n_{\min\{s,f\}}\,,\,\lVert t\rVert=1} |\det(\vB_{t})|}_{=D_{n,\min\{s,f\}}}>0
 \end{equation}
 which shows the claim of the theorem.\qed
 
 \end{proof}

It is important to add here that the compression via the Freiman isomorphism $\phi:I\cup J\rightarrow [n]$ 
is obviously not global and depends on the support sets $I$ and $J$.
From numerical point of view one might therefore proceed only with the first assertion in \eqref{eq:Bty}
and evaluate the particular intermediate steps:
\begin{align}
   \begin{split}
      \inf_{\substack{(\vx,\vy)\in (\Sigma_s,\Sigma_f)\\ \Norm{\vx}=\Norm{\vy}=1}} \Norm{\vx * \vy}^2
     &= \min_{\substack{(\vtx,\vty)\in (\Sigma_s^n,\Sigma_f^n)\\ \Norm{\vtx}=\Norm{\vty}=1}} \Norm{\vtx * \vty}^2\\
     &= \min\Big\{\min_{{\tI\in[n]_s}}\min_{\substack{\vty\in \Sigma^{n}_{f}\\\Norm{\vty}=1}} \lam_{\min}(\vB_{\tI,\vty}) 
     ,\min_{{\tJ\in[n]_f}}\min_{\substack{\vtx\in \Sigma^{n}_{s}\\\Norm{\vtx}=1}} \lam_{\min}(\vB_{\tJ,\vtx})\Big\} \\ 
     &\geq\min_{T\in [n]_{\max\{s,f\}}} \min_{t\in \Sigma^n_{\min\{s,f\}}, \Norm{t}=1} \lambda_{\min} (\vB_{T,t})\\
     &\geq   \min_{\substack{\va\in\Sigma_{\min\{s,f\}}^{n}\\\Norm{\va}=1}}\lam_{\min}(\vB_{\va})
     \geq \min_{\substack{\va\in\C^n\\\Norm{\va}=1}}\lam_{\min}(\vB_{\va}).
 \end{split}\label{eq:bys2}
 \end{align}
 The first equality holds, since any support configuration in $\Sigma_s^n \times \Sigma_f^n$ is also realised by
 sequences in $\Sigma_s \times \Sigma_f$.  The bounds in \eqref{eq:bys2} can be used for numerical computation attempts.

Let us now summarize the implications for the RNMP of zero-padded sparse circular convolutions as
defined in \eqref{eq:defcconv}.
Therefore we denote the zero-padded elements by $\Sigma^{n,n-1}_s:=\{\vx\in
\C^{2n-1}|\supp(\vx) \in [n]_s\}$, for which the circular convolution \eqref{eq:defcconv} equals the ordinary convolution
\eqref{eq:defconv} restricted to $[2n-1]$.  Hence, the  bounds in \thmref{thm:dryi} will be valid also in
\eqref{eq:approx:rnmp:1} for  $X=\Sigma^{n,n-1}_s$ and $Y=\Sigma_f^{n,n-1}$.

\begin{corollary}\label{cor:cryi}
   For $s,f\leq n$ and all $(\vx,\vy)\in \Sigma_s^{n,n-1} \times\Sigma_f^{n,n-1}$
   it holds: 
   \begin{align}
      \alp(s,f,n) \Norm{\vx}\Norm{\vy} \leq \Norm{\vx \circledast \vy} \leq \beta(s,f)\Norm{\vx}\Norm{\vy}
      \label{eq:cryi}.
   \end{align}
   Moreover, we have $\bet^2(s,f)=\min\{s,f\}$ and with $\tn=\min\{n, 2^{2(s+f-2)\log(s+f-2)}\} :$
  \begin{align}
     \alp^2 (s,f,\tn)\geq \frac{1}{\sqrt{\tn\cdot\min(s,f)^{\tn-1}}}\cdot D_{\tn,\min(s,f)},      
  \end{align}
  which is a decreasing sequence in $s$ and $f$. For $\beta(s,f)=1$ we get equality with $\alp(s,f)=1$. 
\end{corollary}
\begin{proof} Since $x\in \Sigma^{n,n-1}_s$ and $y\in\Sigma^{n,n-1}_f$ we have 
  \begin{align}
    \Norm{\vx\circledast\vy}_{\ell^2 ([2n-1])}=\Norm{\vx * \vy}_{\ell^2([-n,n])}.
  \end{align}
  Hence, $x,y$ can be embedded in $\Sigma_s$ resp. $\Sigma_f$ without changing the norms.  If $n\geq 
 \lfloor2^{2(s+f-2)\log(s+f-2)}\rfloor=:\tn$, then we can find a
  Freiman isomorphism which express the convolution by vectors $\vtx,\vty\in \C^{\tn}$. If $n\leq \tn$ there is no need
  to compress the convolution and we can set easily $\tn=n$.  Hence, all involved Hermitian Toeplitz matrices $\vB_t$ in
  \eqref{eq:Bty} are $\tn\times\tn$ matrices and we just have to replace $n$ by $\tn$ in \eqref{eq:lowerbound}. 
\end{proof}

\subsection{Implications for Phase Retrieval}

In this section we will discuss an interesting application of the RNMP result in Theorem \ref{thm:dryi} 
and in particular we will exploit here the version presented in Corollary \ref{cor:cryi}. 
We start with a bilinear map $B(\Inx,\Iny)$ which is \emph{symmetric}, i.e., $B(\Inx,\Iny)=B(\Iny,\Inx)$ and
let us denote its diagonal part by
$A(\Inx):=B(\Inx,\Inx)$. Already in \eqref{eq:model:quadratic} we mentioned \emph{quadratic inverse problems}
where $\Inx\in\Sigma_s$ and there we argued that,
due the binomial-type formula:
\begin{equation}
   \begin{split}
      A(\vx_1)&-A(\vx_2)= B(\vx_1-\vx_2,{\vx_1+\vx_2}) 
   \end{split}
   \label{eq:binom}
\end{equation}
%      &=B(\vx_1,{\vx}_1)-B(\vx_2,{\vx}_2) + B(\vx_1,\vx_2)-B(\vx_1,\vx_2)\\
%
different $\vx_1$ and $\vx_2$ can be (stable)
distinguished modulo global sign on the basis of $A(\vx_1)$ and $A(\vx_2)$ whenever $B(\vx_1-\vx_2,\vx_1+\vx_2)$ is
well-separated from zero. In the sparse case $\vx_1,\vx_2\in\Sigma_s$ this assertion is precisely given by property \eqref{eq:approx:rnmp}
when lifting $B$ to a linear map operating on the set $U=\slmat_{2s,2s}$ of rank--one matrices with at most $2s$ non--zero rows and columns
(see again \eqref{eq:model:quadratic}). In such rank--one cases we call this as the RNMP condition and for sparse convolutions (being symmetric)
we have shown in the previous Section \ref{sec:rnmp:sparse} that this condition is fulfilled independent of the ambient dimension.
As shown in \corref{cor:cryi} this statement translates to zero-padded circular convolutions. 
%$B(\vx_1,\vx_2)=\vx_1\circledast \vx_2$ on 
%$\Sigma_s^{n,n-1}:=\{\vx\in\C^{2n-1}|\supp(\vx) \subset[n]_s\}$
Hence, combining \eqref{eq:binom} with \corref{cor:cryi} and \thmref{lemma:randomsampling:bilinear} asserts that each
zero--padded $s$--sparse $\vx$ can be stable recovered modulo global
sign from $\Omi(s\log n)$  randomized samples of its circular auto-convolution (which itself is at most $s^2-$sparse).

However, here we discuss now another important application for the \emph{phase retrieval problem} and these implications 
will be presented also in \cite{walk:symfourier13}. The relation to the quadratic problems above is as follows:
Let us define from the (symmetric) circular convolution  $\circledast$ the (sesquilinear) \emph{circular correlation}:
\begin{equation}
   \vx\ostar \vy:=\vx \circledast\vGam\overline{\vy}=\Fmatrixa(\Fmatrix \vx \odot\overline{\Fmatrix\vy})
   \label{eq:circcorr}
\end{equation}
where $(u\odot v)_k:=u_kv_k$ denotes the Hadamard (point--wise) product,
$(\Fmatrix)_{k,l}=n^{-\frac{1}{2}}e^{i 2\pi kl/n}$ is the unitary Fourier matrix (here on $\C^n$) and $\vGam:=\Fmatrix^2=\Fmatrixa^2$ 
is the time reversal (an involution). Whenever dimension is important we will indicate this by
$\Fmatrix=\Fmatrix_n$ and $\vGam=\vGam_n$. Therefore, Fourier measurements on the circular auto--correlation $\vx\ostar \vx$ are
intensity measurements on the Fourier transform of $\vx$:
\begin{align}
  \Fmatrix(\vx \ostar \vx)= |\Fmatrix \vx|^2\label{eq:fmatrix}.
\end{align}
Recovering $\vx$ from such intensity measurements is
known as a phase retrieval problem, see e.g. \cite{BCMN13} and references therein, 
which is without further support restrictions on $\vx$ not
possible \cite{Fie87}. Unfortunately, since the circular correlation in \eqref{eq:circcorr} is 
sesquilinear and not symmetric \eqref{eq:binom} does not hold in general. However, it will hold
for structures which are consistent with a real--linear algebra, i.e. \eqref{eq:binom} symmetric for vectors
with the property $\vx=\vGam\bar{\vx}$ (if and only if and the same also for $\vy$). 
Hence, to enforce this symmetry and to apply our result, we perform a \emph{symmetrization}.
Let us consider two cases separately. First, assume that $x_0=\bar{x}_0$ and define
$\Sym\colon \C^n \to \C^{2n-1}$:
\begin{align}
   \Sym( \vx):
   =(\underbrace{x_0,x_1,\dots,x_{n-1}}_{=\vx},\underbrace{\bar{x}_{n-1},\dots,\bar{x}_1}_{=:\vx_-^\circ})^T
   \label{eq:sym}.
\end{align}
Now, for $x_0=\bar{x}_0$ the symmetry condition
$\Sym(\vx)=\vGam\cc{\Sym(\vx)}$ is fulfilled (note that here $\vGam=\vGam_{2n-1}$):
\begin{align}
   \Sym(\vx)=\begin{pmatrix}{\vx}\\ \cc{\vx^\circ_-} \end{pmatrix}
   =\vGam\begin{pmatrix}\cc{\vx}\\ \vx^\circ_- \end{pmatrix}
   =\vGam\cc{\begin{pmatrix}\vx\\\cc{\vx_-^\circ} \end{pmatrix}}
   =\vGam\cc{\Sym(\vx)}.
   \label{eq:involutioninv}
\end{align}
Thus, for $\vx,\vy\in\C^n_0:=\{x\in\C^n\,:\,x_0=\bar{x}_0\}$, circular
correlation of (conjugate) symmetrized vectors is symmetric and agrees with the circular convolution.  
Let us stress the fact, that the symmetrization map is linear only for \emph{real} vectors $\vx$ since complex
conjugation is involved. On the other hand, $\Sym$ can obviously be written as a linear map on vectors like
$(\text{Re}(\vx),\text{Im}(\vx))$ or $(\vx,\bar{\vx})$.\\[.3em]  
\if0
\noi {\bfseries Linear phase filters:} Let us define the circular right-shift $\vS_n:\C^n\to \C^n$ by $\vS_n \vx=
(x_{n-1},x_0,x_1,\dots,x_{n-2})^T$. Then the impulse response
$\vh:=\vS_{2n-1}^{n-1}\Sym(\vx) \in \C^{2n-1}$, defines an
odd-length \emph{linear-phase filter} $H(z)=\sum_{k=0}^{2n-2} h_k z^{-k}$ for $z\in \C$ if $x_0\in\R$, since we have
$\cc{h_0}=h_{2n-2}\not=0$ and
%
\begin{align}
  \cc{h_k} = h_{2n-2 - k}\for k\in \{0,\dots,2n-2\}.
\end{align}
%
If $x_{n-1}\not=0$, then the impulse response or filter is called \emph{Hermitian} or \emph{conjugate symmetric} of
order $2n-2$, see e.g.  \cite[Cha.2]{Vai93}. Hence, by the shift-invariance we get\footnote{ Note, that
  $\vS_{2n-1}^{n-1}$ centers the impulse response  such that it  becomes a causal FIR filter.} for
  $\vx\in\C^n_0:=\{\vx\in \C^n| x_0\in\R\}$
%
\begin{align}
  A(\vx)&= \Sym(\vx)\circledast \Sym(\vx) = \vS_{2n-1}^{n-1}\Sym(\vx)\circledast\vS_{2n-1}^{n-1}\Sym(\vx)=\vh\circledast\vh,
\end{align}
%
which is the circular auto-convolution of a linear-phase filter.\\
\fi
Applying
\corref{cor:cryi} to the \emph{zero-padded symmetrization} 
(first zero padding $n\rightarrow 2n-1$, then symmetrization $2n-1\rightarrow 4n-3$) $\Sym (\vx)$ for
$\vx\in\Sigma^{n,n-1}_{0,n}:=\Sigma^{n,n-1}_{n}\cap\C_0^{2n-1}$ we get the following stability result. 
%
%%%%%%%%%%%%%% THEOREM: Phase Retrieval %%%%%%%%%%%%%%%%%%%%%%%%%%%%%%%%%%%%%%%%%%%%%%%%%%%%%%%%%%%%%%%%%%
\begin{theorem}\label{thm:phaseretrieval}
  Let $n\in\N$, then $4n-3$ absolute-square Fourier measurements of zero padded symmetrized vectors in $\C^{4n-3}$
  are stable up to a global sign for $\vx\in\Sigma^{n,n-1}_{0,n}$, i.e., for all $\vx_1,\vx_2\in \Sigma^{n,n-1}_{0,n}$ it
  holds
  \begin{align}
     \Norm{|\Fmatrix\Sym (\vx_1)|^2- |\Fmatrix\Sym(\vx_2)|^2}
     &\geq  c\Norm{\Sym(\vx_1-\vx_2)}\Norm{\Sym(\vx_1+\vx_2)}
     \label{eq:stablequadratic}
  \end{align}
 with $c=c(n)=\alp(n,n,4n-3)/\sqrt{4n-3}>0$ and $\Fmatrix=\Fmatrix_{4n-3}$.
\end{theorem}
%%%%%%%%%%%%%%%%%%%%%%%%%%%%%%%%%%%%%%%%%%%%%%%%%%%%%%%%%%%%%%%%%%%%%%%%%%%%%%%%%%%%%%%%%%%%%%%%%%%%%%%%%
\noindent
\begin{remark}
Note that we have:
\begin{equation}
   2\lVert \vx\rVert^2\geq \lVert \Sym(\vx)\rVert^2=\lVert \vx\rVert^2+\lVert \vx^\circ_-\rVert^2\geq\lVert \vx\rVert^2.
\end{equation}
Thus, $\Sym(\vx)=0$ if and only if $\vx=0$ and the stability in distinguishing $\vx_1$ and $\vx_2$ up to a global sign
follows from the RHS of \eqref{eq:stablequadratic} and reads explicitly as:
\begin{equation}
   \Norm{|\Fmatrix\Sym (\vx_1)|^2- |\Fmatrix\Sym(\vx_2)|^2}
   \geq  c \Norm{\vx_1-\vx_2}\Norm{\vx_1 + \vx_2}.
\end{equation}
Unfortunately, $s-$sparsity of $x$ does not help in this context to reduce the number of measurements, but at least can enhance the
stability bound $\alp$ to $\alp(2s,2s,4n-3)$.
\end{remark}
\begin{proof}
   For zero-padded symmetrized vectors,  auto-convolution agrees with auto-correlation and
   we get from \eqref{eq:sym} for $\vx\in\Sigma_{0,n}^{n,n-1}$:
   \begin{align}
     \Fmatrix(A(\vx))=\Fmatrix(\Sym(\vx)\circledast\Sym(\vx))=\sqrt{4n-3}\Betrag{\Fmatrix\Sym(\vx)}^2.
   \end{align}
   Putting things together we get for every $\vx\in\Sigma_{0,n}^{n,n-1}$:
   \begin{align}
     \begin{split}
       \big\| |\Fmatrix\Sym (\vx_1)|^2 \! -\!\Betrag{\Fmatrix\Sym(\vx_2)}^2 \big\| &= ({4n-3})^{-1/2}
         \Norm{\Fmatrix(A(\vx_1)-A(\vx_2))}\\
         \text{$\Fmatrix$ \small is unitary}\ra &\,=({4n-3})^{-1/2} \Norm{
         A(\vx_1)-A(\vx_2)}\\
         &\overset{\eqref{eq:binom}}{=}({4n-3})^{-1/2} \Norm{\Sym (\vx_1 -\vx_2) \!\circledast\! \Sym (\vx_1 +\vx_2)}\\ 
         &\geq\!\frac{\alp(n,n,4n-3)}{\sqrt{4n-3}} \Norm{\Sym (\vx_1 -\vx_2)}\cdot\Norm{\Sym (\vx_1 +  \vx_2)}.
       \end{split}
   \end{align}
   In the last step we use that \corref{cor:cryi} applies whenever the non-zero entries are contained in
   a cyclic block of length $2n-1$.
\end{proof}
In the \emph{real case} \eqref{eq:stablequadratic} is equivalent to a \emph{stable linear embedding} in $\R^{4n-3}$ up
to a global sign (see here also \cite{EM12} where the $\ell_1$--norm is used on the left
side) and therefore this is an \emph{explicit phase retrieval statement} for \emph{real} signals. Recently, stable
recovery also in the complex case up to a global phase from the same number of subgaussian measurements has been
achieved in \cite{EFS13} using lifting as in \eqref{eq:model:quadratic}.   
Both results hold with exponential high probability whereby our result is
deterministic. Even more, the greedy algorithm in \cite[Thm.3.1]{EFS13} applies
in our setting once the signals obey sufficient decay in
magnitude.  But, since $\Sym$
is \emph{not complex-linear} Theorem \ref{thm:phaseretrieval} cannot directly be compared with the usual complex phase
retrieval results.  On the other hand, our approach indeed (almost) distinguishes complex phases by the Fourier
measurements since
symmetrization provides injectivity here up to a global sign.  To get rid
of the odd definition $\C_0^n$ one can symmetrize (and zero padding)  $\vx\in\C^n$ also by:
\begin{align}
  \Symz(\vx):=(\underbrace{0,\dots,0}_{n},x_0,\dots, x_{n-1},\bar{x}_{n-1},\dots,\bar{x}_0,\underbrace{0,\dots,0}_{n-1})^T\in\C^{4n-1}\label{eq:s2}
\end{align}
again satisfying $\Symz(\vx)=\vGam_{4n-1}\cc{\Symz(\vx)}$ at the price of two 
further dimensions.
%
%%%%%%%%%%%%%%%%%%%%%%%%%%%%%%%%%%%%%%%%%% 5
\begin{corollary}
   Let $n\in\N$, then $4n-1$ absolute-square Fourier measurements of zero padded and symmetrized vectors given by
   \eqref{eq:s2} are
   stable up to a global sign for $\vx\in\C^n$, i.e., for all $\vx_1,\vx_2\in \C^n$ it holds
     \begin{equation}
      \Norm{|\Fmatrix \Symz (\vx_1)|^2- |\Fmatrix\Symz(\vx_2)|^2}
      \geq  2 c \Norm{\vx_1-\vx_2}\Norm{\vx_1 + \vx_2} 
      \label{eq:stablequadratic2}
   \end{equation}
   with $c=c(n)=\alp(n,n,4n-1)/\sqrt{4n-1} >0$ and $\Fmatrix=\Fmatrix_{4n-1}$.
\end{corollary}
%%%%%%%%%%%%%%%%%%%%%%%%%%%%%%%%%%%%%%%%%%5

The proof of it is along the same steps as in \thmref{thm:phaseretrieval}.
The direct extension to sparse signals as in \cite{WJ12b} seems to be difficult since randomly chosen Fourier 
samples do not provide a sufficient measure of concentration property without further randomization.

\begin{acknowledgement}
   The authors would like to thank the anonymous reviewers for their detailed and 
   valuable comments.
   We also thank Holger Boche, David Gross, Richard Kueng 
   and G\"otz Pfander for
   their support and many helpful discussions.
\end{acknowledgement}
%
%\section*{Appendix}
%\addcontentsline{toc}{section}{Appendix}

\printbibliography
%\bibliographystyle{spbasic}%abbrvnat

% use this for our local bib-files
%\bibliography{library,jabref_philipp_utf2}

%\bibliography{library}

% use this for the final bib-file
%\bibliography{chapter}

\end{document}

%% file: HeaderShamaiElsevier2.tex
\newcommand{\alp}{\ensuremath{\alpha}}

\newcommand{\bet}{\ensuremath{\beta}}

\newcommand{\del}{\ensuremath{\delta}}

\newcommand{\tht}{\ensuremath{\theta}}

\newcommand{\ti}{\ensuremath{\tilde{i}}}
\newcommand{\tI}{\ensuremath{\tilde{I}}}
\newcommand{\tJ}{\ensuremath{\tilde{J}}}
\newcommand{\tj}{\ensuremath{\tilde{j}}}

\newcommand{\gam}{\ensuremath{\gamma}}

\newcommand{\vGam}{\ensuremath{{\boldsymbol\Gamma}}}
\newcommand{\lam}{\ensuremath{\lambda}}

\newcommand{\ome}{\ensuremath{\omega}}

\newcommand{\hx}{\ensuremath{u}}

\newcommand{\tn}{\ensuremath{\tilde{n}}}

\newcommand{\for}{\ensuremath{\quad{\text{for}\quad}}}
\newcommand{\foral}{\ensuremath{\quad{\text{for all}\quad}}}

\newcommand{\Omi}{{\ensuremath{\mathcal{O}}}}

\newcommand{\Ao}{{\ensuremath{\mathcal{A}^\circ}}}
\newcommand{\Aoz}{{\ensuremath{\mathcal{A}^\circ_{z}}}}

\newcommand{\Sym}{{\ensuremath{\mathcal{S}}}}
\newcommand{\Symz}{{\ensuremath{\mathcal{S}'}}}
\newcommand{\Symo}{{\ensuremath{\mathcal{S}^\circ}}}
\newcommand{\Symoz}{{\ensuremath{\mathcal{S}^\circ_{z}}}}

\newcommand{\ko}{{\ensuremath{\mathbb K}}}

\newcommand{\R}{{\ensuremath{\mathbb R}}}
\newcommand{\N}{{\ensuremath{\mathbb N}}}

\newcommand{\Z}{{\ensuremath{\mathbb Z}}}

\newcommand{\Fmatrix}{{\ensuremath{F}}}
\newcommand{\Fmatrixa}{{\ensuremath{F}^*}}

\newcommand{\ra}{\rightarrow}

\newcommand{\LRA}{\ensuremath{\Leftrightarrow} }

\newcommand{\vx}{{\ensuremath{x}}}

\newcommand{\hvx}{{\ensuremath{u}}}

\newcommand{\vy}{\ensuremath{y}}

\newcommand{\hy}{{\ensuremath{v}}}
\newcommand{\hvy}{{\ensuremath{{v}}}}

\newcommand{\va}{{\ensuremath{t}}}

\newcommand{\vby}{{\ensuremath{\tilde{y}}}}

\newcommand{\zero}{{\ensuremath{\mathbf 0}}}

\newcommand{\tx}{\ensuremath{\tilde{x}}}

\newcommand{\ty}{\ensuremath{\tilde{y}}}

\newcommand{\vtx}{\ensuremath{{\tilde{x}}}}
\newcommand{\vty}{\ensuremath{{\tilde{y}}}}

\newcommand{\vB}{{\ensuremath{B}}}

\newcommand{\vS}{\ensuremath{S}}                         % Vektorwertiges Ma{\ss}
\newcommand{\vh}{\ensuremath{h }}                         % Vektorwertige Funktion
\newcommand{\thmref}[1]{Satz~\ref{#1}}

\newcommand{\secref}[1]{Abschnitt~\ref{#1}}

\newcommand{\noi}{\noindent}

\ifx\definition\undefined
\newtheorem{definition}{Definition}         %%%%%%%%%%%%%%%%%%%%%%%%%%%%%%
\fi
\ifx \@definition \@empty
\newtheorem{definition}[theorem]{Definition}   % the section where they     %
\fi

\ifx\conjecture\undefined
\newtheorem{conjecture}{Conjecture}         %%%%%%%%%%%%%%%%%%%%%%%%%%%%%%
\fi
\ifx\theorem\undefined
\newtheorem{theorem}{Theorem}         %%%%%%%%%%%%%%%%%%%%%%%%%%%%%%
\fi
\ifx\lemma\undefined
\newtheorem{lemma}{Lemma}
\fi
\ifx\question\undefined
\newtheorem{question}{\color{red}{Question}}         %%%%%%%%%%%%%%%%%%%%%%%%%%%%%%
\fi
\ifx\proposition\undefined
\newtheorem{proposition}[theorem]{Proposition} %%%%%%%%%%%%%%%%%%%%%%%%%%%%%% 
\fi
{%\endgraf%
\comment\noi}%
{\endcomment}

{\par\noindent{\em Beweis\/}.}%
{\hspace*{\fill}{\qed}\vspace{1ex}\par}
{\par\noindent{\em Proof\/}.}%
{\par}

{\hspace*{\fill}{}\vspace{1ex}\par}
{\par\vspace{1.5ex}\noindent{\em Remark\/}.}
{\par\vspace{1.5ex}}
{\noi\vspace{0.5ex}\small}
{\vspace{0.5ex}\par\normalsize}

\newcounter{Examplecount}
{\color{gray}}
{\vspace{0.5ex}\par\normalsize}

{\renewcommand{\labelenumi}{(\roman{enumi})}\begin{list}{\labelenumi}
{\usecounter{enumi}\setlength{\labelwidth}{1.5cm}\setlength{\topsep}{0.3cm}\setlength{\itemsep}{-3pt}}}
{\end{list}}
{\renewcommand{\labelenumi}{(\arabic{enumi})}\begin{list}{\labelenumi}
{\usecounter{enumi}\setlength{\labelwidth}{1.5cm}\setlength{\topsep}{0.3cm}\setlength{\itemsep}{-3pt}}}
{\end{list}}
{\renewcommand{\labelenumi}{$\bullet$}\begin{list}{\labelenumi}
{\setlength{\labelwidth}{1.5cm}\setlength{\topsep}{0.3cm}\setlength{\itemsep}{-2pt}}}
{\end{list}}

\makeatletter
\newcommand{\set}[2]{\ensuremath{%
\setbox0=\hbox{\ensuremath{#2}}
\dimen@\ht0
\advance\dimen@ by \dp0
\left\{\left.#1\rule[-\dp0]{0pt}{\dimen@}\;\right|\;#2\right\} }}
\makeatother

\newcommand\diam{\operatorname{diam}}

\newcommand{\Betrag}[1]{\ensuremath{ \left|#1\right| }}

\newcommand{\Norm}[1]{\ensuremath{ \left\|#1\right\| }}

\newcommand{\skprod}[1]{\ensuremath{ \left\langle #1 \right\rangle }}

\newcommand{\cc}[1]{{\ensuremath{\overline{#1}}}} % complex conjugation

\newcommand{\namen}[1]{{\textsc{#1}}}           % englische w{\"o}rter als kursiv schreiben
\ifx \@paragraph \@empty
\makeatletter
\renewcommand\paragraph{\@startsection
{paragraph}{4}{\z@}{-3.5ex plus-1ex minus-.2ex}%
{1.3ex plus.2ex}{\normalfont\itshape}}

\fi
\renewcommand{\thmref}[1]{Theorem~\ref{#1}}
\renewcommand{\secref}[1]{Section~\ref{#1}}

\newcommand{\corref}[1]{Corollary~\ref{#1}}